\documentclass[onecolumn,11pt,draftcls]{IEEEtran}
\usepackage{epsfig,amsmath,amssymb,epsf,cite, amsthm, scalefnt, subfig}
\newcommand{\argmax}{\mathop{\mathrm{argmax}}}

\newtheorem{lemma}{Lemma}

  \def\cC{{\mathcal{C}}}
  \def\cF{{\mathcal{F}}}

 \def\cN{{\mathcal{N}}}

 \def\argmax{\mathop{\mathrm{argmax}}}

 \def\btheta{{\pmb{\theta}}}
  \def\b0{{\pmb{0}}}

  \def\bc{{\mathbf{c}}}
  \def\bff{{\mathbf{f}}}
\def\bg{{\mathbf{g}}} \def\bh{{\mathbf{h}}} 
  
\def\bm{{\mathbf{m}}} \def\bn{{\mathbf{n}}} 
 \def\bq{{\mathbf{q}}} \def\br{{\mathbf{r}}}
\def\bs{{\mathbf{s}}}  
 \def\bw{{\mathbf{w}}} \def\bx{{\mathbf{x}}}
\def\by{{\mathbf{y}}} \def\bz{{\mathbf{z}}}

  \def\bC{{\mathbf{C}}}
\def\bD{{\mathbf{D}}}  \def\bF{{\mathbf{F}}}
 \def\bH{{\mathbf{H}}} \def\bI{{\mathbf{I}}}

  \def\bU{{\mathbf{U}}}
\def\bV{{\mathbf{V}}}  \def\bX{{\mathbf{X}}}
\def\bY{{\mathbf{Y}}} \def\bZ{{\mathbf{Z}}} 
\def\bLambda{{\mathbf{\Lambda}}}

\def\sc{{\boldsymbol{\mathsf{c}}}} 
 \def\sf{{\boldsymbol{\mathsf{f}}}}

 \def\sl{{\boldsymbol{\mathsf{l}}}}

  \def\mC{{\mathbb{C}}}

\begin{document}
\linespread{1.065}
\scalefont{.91}

\title{Using Channel Output Feedback to Increase Throughput in Hybrid-ARQ
\thanks{The material in this paper was presented in part at the \textit{IEEE
Military Communications Conference}, San Jose, CA, November 2010 and the \textit{IEEE
International Workshop on Signal Processing Advances in Wireless Communications}, San Francisco, CA, June 2011.
The authors are with the School of Electrical and Computer Engineering, Purdue University,
West Lafayette, IN 47907 USA (e-mail: magrawal@purdue.edu;
zchance@purdue.edu; djlove@ecn.purdue.edu; ragu@ecn.purdue.edu).}}
\author{Mayur Agrawal, Zachary Chance, David J. Love, and Venkataramanan Balakrishnan}

\maketitle

\begin{abstract}
Hybrid-ARQ protocols have become common in many packet transmission
systems due to their incorporation in various standards. Hybrid-ARQ combines
the normal automatic repeat request (ARQ) method with error correction codes to
increase reliability and throughput.  In this paper, we look at improving upon
this performance using feedback information from the destination, in particular,
using a powerful forward error correction~(FEC) code in conjunction with a
proposed linear feedback code for the Rayleigh block fading channels. The new
hybrid-ARQ scheme is initially developed for full received packet feedback in a
point-to-point link. It is then extended to various different multiple-antenna
scenarios (MISO/MIMO) with varying amounts of packet feedback information.
Simulations illustrate gains in throughput.
\end{abstract}


\section{Introduction}

The tremendous growth in demand for throughput in wireless networks warrants new design principles for
coding information at the lower layers. In recent years, packet-based
hybrid automatic repeat request (ARQ), which integrates forward error
correction~(FEC) coding with the traditional automatic repeat request protocol, has sparked
much interest. Any hybrid-ARQ scheme includes the transmission of an
acknowledgement~(ACK) or a negative-acknowledgement~(NACK) from the destination
to the source. Although not normally viewed this way, the feedback of ACK/NACK
can be seen as a form of \emph{channel output information~(COI)} indicating to
the source the `quality' of the channel output. Exploiting the
full potential of COI at the source for hybrid-ARQ schemes, however, has not
been explored in the literature. In fact most of the discussion about feedback
in wireless systems has been limited to the use of \emph{channel state
information~(CSI)} at the source \cite{love2,SuMa06, SaDi06}.

Research since the 1960s~\cite{Schal1, Schal2, butman} has long established the
utility of using COI at the source to increase reliability in
additive white Gaussian noise~(AWGN) channels. The gain in reliability scales very fast
in blocklength and relies on simple, linear coding schemes. These schemes
achieve a doubly exponential decay in the probability of error as a function of
the number of packet retransmissions; this is in stark contrast to open loop
systems~(without COI) that can achieve only singly exponential decay in
the probability of error. Most of the the literature for COI assumes an
information-theoretic perspective for analysis. In this work, however, we take a
more signal processing approach to COI. In particular, we explore the
efficacy of including COI to increase the throughput of a practical
hybrid-ARQ scheme.

As is commonly done in hybrid-ARQ literature, we assume that the transmissions
will take place over a \emph{block fading channel} where the fading
characteristics of the channel are assumed to be constant over the transmission
of a packet (e.g., \cite{Caire1, Caire3, Caire 4, Love1}).  In the context of
this channel, there will be two main types of side-information that we consider
to be available at the source: CSI and COI. CSI is commonly known as a
complex-valued quantity that represents the spatial alignment of the channel.
This can be either outdated, where the source has access to outdated values of
CSI (i.e., the channel states for previous blocks), or current, where the source
has access to the present value of CSI (i.e., the channel state for the current
block). Outdated CSI is commonly obtained through feedback techniques from the
destination. Current CSI can be obtained through numerous methods such as
exploiting channel reciprocity through time division duplexing. In fact, most recent 
standards \cite{bruno} and technologies like multi-user MIMO, network
MIMO, and OFDM incur major penalties in performance without the availability of
current CSI~\cite{SpPe04, StBa04, PaDa08}. The most common form of COI is simply
the past received packet at the destination; we will be using this
definition as the destination incurs no processing before feeding back the
received packet. The main focus of the paper will be the integration of
COI in to a hybrid-ARQ framework.

Hybrid-ARQ improves the reliability of the transmission link by jointly
encoding or decoding the information symbols across multiple received packets.
Specifically, there are three ways \cite{Lin} in which hybrid-ARQ schemes are implemented:
\begin{itemize}
  \item \emph{Type I}: Packets are encoded using a fixed-rate FEC code, and
  both information and parity symbols are sent to the destination. In
  the event that the destination is not able to decode the packet, it
  rejects~(NACK) the current transmission and requests the retransmission of the same
  packet from the source. Subsequent retransmissions from the
  source are merely a repetition of the first transmission.
  \item \emph{Type II}: In this case, the destination has a buffer to
  store previous unsuccessfully transmitted packets. The first packet sent
  consists of the FEC code and each subsequent retransmission consists of only
  the parity bits~(\emph{incremental redundancy}) to help the receiver at the
  destination jointly decode across many retransmissions of the same packet.

  \item \emph{Type III}: This method is similar to Type II with one major
  difference. In Type III, every retransmission is self decodable, e.g.,
  \emph{Chase combining}\cite{Chase1}. Therefore the destination has the
  flexibility to either combine the current retransmission with all the
  previously received retransmissions or use only the current packet for
  decoding.
\end{itemize}

The first mention of hybrid-ARQ techniques can be traced back to papers from the
1960s (e.g., \cite{WozHor2}).  However, most attention to this
protocol has been given during the late 1990s and early 2000s.  Throughput and
delay analyses were done for the Gaussian collision channel in
\cite{Caire1,Caire2,Caire3,Caire4}.  These topics were also investigated for
wireless multicast in \cite{Love1} and for block fading channels with modulation
constraints in \cite{Val}. The hybrid-ARQ technique has been looked at when
using many different types of FEC codes including turbo codes
\cite{Jung, Rowitch, Aci}, convolutional codes \cite{Caire3,Hag}, LDPC codes
\cite{Caire4,Ha}, and Raptor codes \cite{Varn,Lee}.  The performance analysis of
a hybrid-ARQ scheme with CSI\cite{SuMa06, SaDi06, SuDi06, SuDi07, KiKa09} and
without CSI at the source\cite{AsPo10} has also been studied in the literature.
In addition, different ways to utilize the feedback channel have been
investigated in \cite{Matsu,Shea}.

%   This paper addresses that problem - how can we use the
% feedback channel efficiently to best increase the reliability of the hybrid-ARQ
% scheme?
In this work, however, we explore the advantages of combining conventional CSI
feedback \emph{with COI feedback}.  As we noted earlier,
the potential of COI feedback in hybrid-ARQ has not been properly explored. The
method we introduce is a variation of Type III hybrid-ARQ scheme that
incorporates the use of COI feedback. In the event that no COI feedback is
available, our proposed scheme simply reduces to the regular Chase combining in
which packets are repeated for retransmission and the destination combines the
received packets using \emph{maximum ratio combining} (MRC).  However when COI
feedback is available, we look at implementing a linear feedback code that is a
generalization of Chase combining to increase the performance of the packet transmission system.  A linear
feedback code is simply a transmission scheme in which the transmit value is a
strictly linear function of the message to be sent and the feedback
side-information \cite{ZaDa11}. We show that such codes provide advantages over
merely repeating the last packet, while offering simpler analysis and
implementation than conventional Type II incremental redundancy codes.

A relevant concern for the implementation of COI feedback techniques
is practicality as it requires (possibly) sending large amounts of data from
the destination back to the source.  However, we hope to address these
concerns by first studying the ideal scenarios (i.e., perfect COI feedback by
feeding back the full received packet) to illustrate what is \textit{theoretically} possible and
then extend these results to limited-resource cases (i.e., noisy COI feedback
and feeding back only parts of the received packet).  This allows us to
establish a trade-off in performance to allow for practical limitations on the
system. Furthermore, COI feedback techniques can be especially beneficial when
there is link asymmetry between the source and the destination. In other words,
the situations in which the reverse link can support much higher rates than
the forward link.

To accommodate the use of multiple-input single-output (MISO) and multiple-input
multiple-output~(MIMO) systems, we first construct the proposed scheme for the
simplest case of single-input single-output (SISO) transmission and then extend
the scheme to the case with multiple transmit antennas. Specifically, the scheme
is adapted for use with MISO and MIMO when current CSI is available at the
source and either perfect or noisy COI feedback is available.  It is also
adapted for MIMO when perfect COI and only outdated CSI is available at the source.

The paper is structured as follows.  In Section \ref{sec:sm}, a brief high-level
description of hybrid-ARQ is given to motivate the investigation into using more
feedback in a packet retransmission scheme.  In Section \ref{sec:lfc}, the
feedback scheme to be integrated into a hybrid-ARQ protocol is introduced for
SISO systems. We begin the section by describing the encoding process. It is
followed by a discussion of decoding; this involves two different cases -
systems with noiseless COI feedback and systems with noisy COI feedback. In
Section \ref{sec:mas}, the SISO scheme is extended to various multiple antenna
scenarios. In Section \ref{sec:harq}, the overall hybrid-ARQ system is discussed
in detail where now the COI feedback schemes created are integrated as a
generalization of Chase combining. Schemes that vary the amount of COI
feedback being sent to the source are also discussed.  In Section \ref{sec:sim},
throughput simulations are given to illustrate the performance of the proposed
hybrid-ARQ scheme versus other commonly used hybrid-ARQ schemes such as
\cite{HSDPA} and traditional Chase combining. Note that our comparison is with
the incremental redundancy ARQ scheme in \cite{HSDPA}~(not the standardized
system in general). It is not our intention to assume the same conditions
present in the transport channels as the ones discussed in \cite{HSDPA}.

\subsection*{Notation:}

The vectors~(matrices) are represented by lower~(upper) boldface letters while
scalars are represented by lower italicized letters. The operators $(\cdot)^T,
(\cdot)^*, \textrm{tr}(\cdot)$, and $\lVert \cdot \rVert$ denote the transpose,
conjugate transpose, trace, and Euclidean norm of a matrix/vector respectively.
The expectation of a random variable or matrix/vector is denoted by $E[\cdot]$.
The boldface letter $\bI$ represents the identity matrix.

\section{System Model}\label{sec:sm}

Consider using the SISO hybrid-ARQ transmission system in Fig. \ref{hybrid}
where there is one antenna available at the source and the destination. The goal
of the transmission scheme is to successfully send the binary information
packet, $\bm \in GF(2)^{L_{\rm{info}}}$, to the destination over a maximum of
$N$ packet retransmissions. $GF(2)$ denotes the Galois field with just two elements
$\{0,1\}$, and $L_{\rm{info}}$ denotes the total number of information bits.
Transmission is accomplished by first encoding the information packet using a
rate $L_{\rm{info}}/L_{\rm{coded}}~(\textrm{where } L_{\rm{info}} \leq
L_{\rm{coded}})$ FEC code, producing a binary codeword of length
$L_{\rm{coded}}$ referred to as $\bc \in GF(2)^{L_{\rm{coded}}}$.  The codeword
is then modulated using a source constellation $\Theta[N]$ (e.g., QAM, QPSK,
etc.) to create a length $L$ packet of modulation symbols called $\btheta \in
\mathbb{C}^{1\times L}$.  Note that the source constellation $\Theta[N]$ is a
function of the maximum number of transmissions $N$. If $N$ is large for a fixed $L$ (i.e., $L/N$ is small), we
may decide to choose a denser constellation to achieve higher throughput. This
is then processed by a packet encoder that encapsulates most of the hybrid-ARQ
process. At this stage, the modulated symbols are further encoded
to generate the transmitted signal $\bx[k] \in \mathbb{C}^{1 \times L}$. It is worthwhile to contrast that
$\btheta$ is the packet of desired information symbols and the entries of
$\bx[k]$ are the actual signals sent at each channel use to convey that
information to the destination. Note that some quantities have a retransmission
index, $k$, which refers to time on the packet level (i.e., for each $k$ a
length $L$ signal, $\bx[k]$, is transmitted).  Furthermore, the transmit vector
is constrained by the power constraint at the source given by
\begin{equation}
\label{pow_const}
E\left[|\bx[k]|^2\right] \leq L\rho, \quad k =
1,\ldots, N,
\end{equation}
\noindent where, as mentioned, $N$ is the maximum number of retransmissions
allowed and $\rho > 0$ is the average power per channel use.

\begin{figure}
  \centering
  \includegraphics[scale = 0.6]{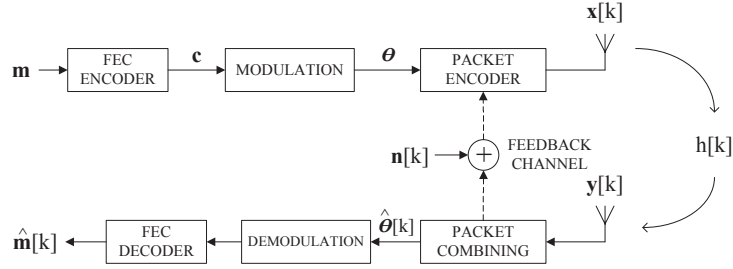}\\
  \caption{The hybrid-ARQ transmission system}\label{hybrid}
\end{figure}

At the destination, the $k^{th}$ retransmission received signal, $\by[k] \in
\mathbb{C}^{1 \times L}$, is obtained.  Using this setup, $\by[k]$ can be written as
\begin{equation}
\label{sys1}
\by[k] = h[k]\bx[k] + \bz[k], \quad 1 \leq k \leq N,
\end{equation}
\noindent where $\bz[k] \in \mathbb{C}^{1\times L}$ is additive noise whose
entries are i.i.d. complex Gaussian such that $\bz[k] \sim \cC\cN(0,\bI)$, and
$h[k] \in \mathbb{C}$ is a zero-mean complex Gaussian random variable with unit
variance. Therefore, we assume that the retransmission takes place across a
Rayleigh block fading channel with blocklength $L$. Note that $\bx, \by,$ and
$\bz$ have been defined as row vectors; this is to aid the later extension of
the scheme to a MISO/MIMO setting. After retransmission $k$, the received packet
is combined using all previously received packets to create an estimate of the
original modulated information packet, $\widehat{\btheta}[k]$. The combining
stage, in Chase combining for example, combines all the received realizations
for a given symbol using MRC. Improving upon the
encoding and combining steps using COI forms the main thrust of this paper; this
will be discussed in detail in the next section. It is worth pointing out that
the incorporation of COI feedback into the hybrid-ARQ scheme is being
implemented at the physical layer. After combining at the destination, the packet is then
demodulated using either soft or hard decoding methods and then passed to the
FEC decoder which then outputs a final estimate of the original information
packet, $\widehat{\bm}[k]$.

It is important to note that a feedback channel is present between the
destination and the source. In fact, for any ARQ protocol, a feedback
channel is necessary so the destination can send back an ACK/NACK signal.
In our setup, we assume that:
\begin{itemize}
\item The destination does not only send back
ACK/NACK information but also
feeds back CSI which could be outdated, current, or
quantized.
\item The destination can feed back the COI for the
packet to the source where \emph{COI feedback} is simply the destination feeding
back exactly (or a subset of) what it has received. This is discussed further in
Section V.
\end{itemize}

Explicitly, the causal COI at the source is equivalent to the source having
access to the past values of $\by[k]$.  However, since \emph{noisy COI feedback}
is also investigated, we introduce a feedback noise process $\bn[k]$ (see Fig.
\ref{hybrid}) so that the source now only has access to past values of
$\by[k] + \bn[k]$.  Note that the source might have access to all or only some
of the entries in $\by[k] + \bn[k]$ based on how much COI is being fed back.
Furthermore, the source can subtract out what it sent due to the availability of $h[k]$. 
Therefore, this is analogous to having access to past values of $\bz[k] + \bn[k]$. The feedback noise,
$\bn[k]$, is assumed to be complex AWGN such that $\bn[k] \sim
\cC\cN(0,\sigma^2\bI)$ and also independent of the forward noise process,
$\bz[k]$.  Note that setting $\sigma^2 = 0$ yields \emph{perfect COI feedback}
as a special case.

\section{Linear Feedback Combining}\label{sec:lfc}

We now narrow our focus to the packet encoding/combining steps of the hybrid-ARQ
system; the full system including the FEC will be considered in Section V.
Specifically, we consider employing COI feedback to better refine the
destination's packet estimate $\widehat{\btheta}[k]$ after each retransmission.
Improving the quality of the estimate will lead to fewer decoding errors and
higher throughput. To begin, we look at the most straightforward setup of SISO,
where the source and destination each have one antenna and outdated CSI along
with causal COI, whether it be noisy or perfect.  The scheme will be extended
for use with multiple antenna scenarios and the effects of varying the CSI
feedback will be discussed in Section IV.

\subsection{Overview}

To construct the new transmission strategy, we aim to develop a linear coding
scheme with the objective of maximizing the post-processed signal-to-noise
ratio~(SNR) after $N$ retransmissions. Initially, we focus on sending only one
symbol or, in other words, assume that $L = 1$ where the information packet $\btheta$ is now a
scalar, $\theta \in \mathbb{C}$.  Note that this assumption is only done to
reduce the amount of notation---the proposed scheme can be readily extended to
arbitrary packet lengths as it is assumed that it will be utilized, in general,
for $L \gg 1$. In the case that $L = 1$, however, the transmit and received
vectors $\bx[k]$ and $\by[k]$ also reduce to scalars $x[k]$ and $y[k]$.

It is helpful at this stage to introduce a mathematical framework for linear
feedback coding.  It can be seen that if $L = 1$, then, gathering all packet
transmissions together, (\ref{sys1}) can be rewritten as
\begin{equation}
\by = \bD\bx + \bz,\label{sys}
\end{equation}
\noindent where $\by = \left[y[1],y[2],\ldots,y[N]\right]^T$ is a column vector
(likewise for $\bx$ and $\bz$) and $\bD = \mathrm{diag}(h[1],h[2],\ldots,h[N])$
is a matrix formed with the channel coefficients down the diagonal.  Note that
the notation $\bD$ is chosen to give distinction between it and the
commonly-used $\bH$ for a MIMO channel matrix which is used later in the paper.
With this setup, we can write the transmit vector $\bx$ as
\begin{equation}
\bx = \bg\theta + \bF(\bz + \bn),\label{trans}
\end{equation}
\noindent where $\bg \in \mathbb{C}^{N \times 1}$ is the vector used to encode
the symbol to be sent, $\theta$, and $\bF \in \mathbb{C}^{N \times N}$ is a
strictly lower triangular matrix used to encode the side-information $\{\bz +
\bn\}$. The form of $\bF$ is constrained to be strictly lower triangular to
enforce causality.  Note that (\ref{trans}) is the transmit structure of linear
feedback coding---the transmitted value is a linear function of the
side-information and of the information message.  Furthermore, the encoding process is
encapsulated by the matrix $\bF$ and the vector $\bg$.

Now, we shift focus to the destination.  After receiving a packet, the
destination forms an estimate of the packet by
\begin{equation}
\widehat{\theta}[k] = \bq_{k}^*\by_{(k)},
\end{equation}
\noindent where $\widehat{\theta}[k]$ is the destination's estimate of the symbol $\theta$ after
$k$ retransmissions, $\bq_{k} \in \mathbb{C}^{N\times 1}$ is the
\emph{combining vector} used for the $k^{th}$ retransmission, and the notation $\by_{(k)}$ refers to the first $k$
entries of $\by$.  It can be seen that the packet estimation process can be
completely described by the vectors $\bq_{k}$.  Thus, the entire linear feedback code
can be represented by the tuple $(\bg,\bF,\bq_{k})$ \cite{ZaDa11}.  Thus, with this framework, the
post-processed SNR for the system after $N$ retransmissions can be defined as
\begin{equation}
SNR = \frac{|\bq^*\bD\bg|^2 \rho}{\|\bq^*(\bI+\bD\bF)\|^2 + \sigma^2\|\bq^*\bD\bF\|^2}.
\end{equation}
\noindent where $\bq = \bq_{N}$; the subscript is dropped for convenience.  We can now mathematically define the overall objective of this
section: to construct a linear feedback code that maximizes the post-processed
SNR given the channel coefficients $h[k]$; this is equivalent to finding
\begin{equation}
(\bg,\bF,\bq)_{opt} = \argmax_{(\bg,\bF,\bq)} \frac{|\bq^*\bD\bg|^2 \rho}{\|\bq^*(\bI+\bD\bF)\|^2 + \sigma^2\|\bq^*\bD\bF\|^2},\label{opt}
\end{equation}
\noindent while satisfying the average power constraint (\ref{pow_const}) and causality of side-information.

In the special case of a SISO system with perfect COI ($\sigma^2 = 0$) and
causal CSI, it can be shown that the solution to (\ref{opt}) has the structure
of the scheme given in \cite{liu1} (the specific derivation is omitted due to
space concerns).  Interestingly, the scheme in that work was derived using a
control-theoretic approach instead of using post-processed SNR as an objective
function \cite{El04}.   Most importantly, the optimality of the scheme in the
SNR sense motivates our construction. In particular, we develop a generalization
of the feedback scheme presented in \cite{liu1} as this scheme was not only
shown to achieve capacity but also achieve a doubly exponential decay in
probability of error. In the proposed generalized scheme, we extend the original
scheme for use with:
\begin{itemize}
\item single and multiple antennas (i.e., MISO and MIMO wireless systems),
\item perfect and noisy COI feedback,
\item outdated and current CSI at the source.
\end{itemize}
\noindent The details of the proposed scheme are given in the following sections.

\subsection{Encoding}

The fundamental idea of the transmission scheme is to transmit the scaled
estimation error from the previous transmission for each successive
retransmission so that the destination can attempt to correct its current estimate \cite{elias,
gallager}.  The scaling operation is performed so the transmitted signal meets
the average power constraint (\ref{pow_const}).  To further illustrate this
concept and help motivate our construction, we now briefly present a heuristic
overview of the scheme in \cite{liu1}.  In this case, the transmitted signal,
$x[k]$, is given as
\begin{equation}
x[k+1] = \delta[k] e[k],\label{liuenc}
\end{equation}
\noindent where $e[k] = \theta - \widehat{\theta}[k]$ is the error in the
destination's estimate of the message after the $k^{th}$ packet reception and
$\delta[k]$ is the scaling factor chosen to appease the power constraint.
After receiving $y[k]$, the destination then forms an estimate of the error,
$\widehat{e}[k]$. This is then subtracted from the current estimate.  As will be
shown, our proposed scheme is motivated by this error-scaling technique.

We now define how the source encodes the message.  As the encoding process for a
perfect COI feedback and the encoding process for a noisy COI feedback are very
similar, we now introduce the encoding process for both perfect and noisy COI
feedback in a single framework.  The encoding operation of the proposed scheme
can be written compactly in the definitions of $\bF$ and $\bg$; they are
constructed as:
\begin{itemize}
\item The $(i,j)^{th}$ entry of $\bF$, $f_{i,j}$, is
\[
f_{i,j} = \left\{\begin{array}{l r}
-\sqrt{\gamma}\rho\phi[i-1]h^*[j], & i>j,\\
0,& i \leq j,\end{array}\right.
\]
\item The $i^{th}$ entry of $\bg$, $g_{i}$, is
\[
g_{i} = \phi[i-1],
\]
\end{itemize}
\noindent where
\[
\phi[k] = \left\{\begin{array}{l r}
\displaystyle\prod_{i=1}^{k}\beta_{(\gamma,\sigma^2)}[i], & k > 0\\
1, & k = 0,\\
\end{array} \right.
\]
\begin{equation}
\beta_{(\gamma,\sigma^2)}[k] = \left(1 + (1+\sigma^2)\gamma\rho |h[k]|^2\right)^{-1/2},\label{betaeq}
\end{equation}
\noindent and $\gamma \in [0,1]$ is a constant.  Note that the scaling factor
$\delta[k]$ in (\ref{liuenc}) is now given its analog by the term $\phi[k]$
which ensures the proposed scheme meets the power constraint (\ref{pow_const}).

The scheme presented here in the form of $\bg$ and $\bF$ is a direct
generalization of the error-scaling scheme in (\ref{liuenc}) as the original
scheme for perfect COI feedback can be obtained as a special case of these
definitions by letting $\gamma = 1$ and $\sigma^2 = 0$.  The main mechanism
introduced into the proposed scheme is a power allocation variable, $\gamma$, to
help combat the effect of the feedback noise, $n[k]$ \cite{ZaDa11}.  Specifically, $\gamma$ is
a degree of freedom introduced to allocate power between the encoding of
feedback side-information and the information to be sent.  It is only of use
when the feedback channel is noisy; if feedback noise is not present, it should
be set to $\gamma = 1$ and disregarded. In brief, as $\gamma \rightarrow 0$,
this scheme simply repeats the packet on every retransmission (i.e., the scheme
becomes equivalent to Chase combining). As $\gamma$ grows, the scheme uses most
of the feedback power to mitigate the noise in the destination estimate. This
quantity is discussed in detail later in the paper.

Now, with $\bg$ and $\bF$ defined, the encoding process is completely described,
and we can now move on to verifying that it meets the average transmit power
constraint (\ref{pow_const}). As will be shown, it is much easier to derive the
average transmit power of the proposed scheme if it is rewritten in a recursive
manner; thus, its recursive form is now presented.  Assuming that the symbol is
scaled such that $E[|\theta|^2] = \rho$, the first packet transmission is set to
the symbol itself with $x[1] = \theta$.  The subsequent transmissions can be
written as
\begin{equation}
x[k + 1] = \beta_{(\gamma,\sigma^2)}[k]\left(x[k] - \sqrt{\gamma}\rho
h^*[k](z[k] + n[k]) \right),~1 < k \leq N.
\label{recur}
\end{equation}
\noindent With the recursive formulation given, we can now present the following lemma.
\begin{lemma}
The proposed scheme in (\ref{recur}) meets the average transmit power constraint
given in (\ref{pow_const}) for both noisy and perfect COI feedback.
\end{lemma}
\begin{proof}
The proof is based on a simple inductive argument. Since the symbol has been
scaled to have a second moment of $\rho$, the average power of the first
transmission is $E[|x[1]|^2] = \rho$.  Assume that $E[|x[k]|^2] =
\rho$ for some $k$. Using (\ref{recur}), we can write the average
transmit power for the $(k + 1)^{th}$ retransmission of packet $\theta$
conditioned on channel realization $h[k]$ as
\begin{eqnarray*}
E\left[|x[k + 1]|^2 \Big{|} h[k]\right] & = &
E\left[\left|\beta_{(\gamma,\sigma^2)}[k]\left(x[k] - \sqrt{\gamma}\rho
h^*[k](z[k] + n[k]) \right)\right|^2\Big{|}h[k]\right]\\
& = & \frac{1}{1 +
(1+\sigma^2)\gamma \rho |h[k]|^2 } E\left[\left|x[k] - \sqrt{\gamma}\rho
h^*[k](z[k] + n[k]) \right|^2\Big{|}h[k]\right]\\
& = & \frac{1}{1 + (1+\sigma^2)\gamma \rho |h[k]|^2 }
\left(E\left[|x[k]|^2\Big{|}h[k]\right] + (1+\sigma^2)\gamma\rho^2
|h[k]|^2\right)\\ & \stackrel{(a)}{=} & \frac{1}{1 + (1+\sigma^2)\gamma\rho |h[k]|^2 }
\left(\rho + (1+\sigma^2)\gamma\rho^2 |h[k]|^2\right)\\
& = & \rho,
\end{eqnarray*}
where the equality in $(a)$ follows from $E\left[|x[k]|^2\Big{|}h[k]\right] = E\left[|x[k]|^2\right] =
\rho$. Therefore by the principle of mathematical induction, the equality holds for any
arbitrary $k$.
\end{proof}

Now that the encoding operation has been described and verified to meet
the average transmit power constraint, it is possible to move on to the decoding
stage.

\subsection{Decoding}
In this section, we discuss the decoding process in the proposed scheme.  It is
worthwhile to point out that we only perform soft signal-level decoding---the
output of the destination is an estimate that is not necessarily mapped to an
output alphabet. Unlike the encoding operation, decoding at the destination
significantly differs depending on whether perfect or noisy COI is available at
the source.

\subsubsection{Perfect COI Decoding ($\sigma^2 = 0$)}
First, we look into defining $\bq$ for
perfect COI. In the special case when the feedback channel is perfect, this scheme assumes
the structure of the feedback scheme in \cite{liu1}; we reproduce it in this
section for completeness.  In this case, the combining vector $\bq$ has
a concise closed form.  In particular using the definition in (\ref{betaeq}),
the $i^{th}$ component of $\bq$, $q_{i}$ can be given as
\begin{equation}
q_{i} = \phi[i-1]\beta_{(1,0)}^{2}[i]\rho h^*[i].
\end{equation}

\noindent Because of this definition, $\bq_{k}$ for perfect COI can be defined as
$\bq_{k} = [q_{1},\ldots,q_{k}]^T$.  Note that since the COI at the source is assumed to be
perfect, $\sigma^2 = 0$ and $\gamma = 1$.  Now that $\bq$ has been defined, the entire scheme for perfect
COI feedback has been described, and the structure of $\bq$ can be
used to formulate the decoding process in a recursive fashion.  Thus, at this point, we
introduce the following lemma:

\begin{lemma}\label{lem2}
The coding scheme for perfect COI feedback can be alternatively represented as
\begin{eqnarray}
\label{recenc}
x[k + 1] & = & \beta_{(1,0)}[k]\left(x[k] -
\rho h^*[k]z[k]\right)\\
\label{recdec} \widehat{{\theta}}[k] & = & \left(1 -
|\phi[k]|^2 \right)\theta + \rho
|\phi[k]|^2\sum_{m =
1}^{k}\phi^{-1}[m - 1]h^*[m]z[m].
\end{eqnarray}
\end{lemma}
\noindent The proof has been relegated to the Appendix.  Note that Lemma 2 suggests that the estimator of the proposed scheme is
a biased one. However, we can easily make the final estimated output unbiased by
performing the appropriate scaling. We can define the unbiased
estimator of packet $\theta$ as
\begin{eqnarray}
\nonumber \widehat{{\theta}}^u[k] & = & \left(1 -
|\phi[k]|^2
\right)^{-1}\widehat{{\theta}}[k]\\
\label{theta_ub} & = & \theta + \rho \left(1 -
|\phi[k]|^2
\right)^{-1} |\phi[k]|^2\sum_{m =
1}^{k}\phi^{-1}[m - 1]h^*[m]z[m].
\end{eqnarray}

% The capacity achieving nature of the above scheme with very good reliability
% has been reported in \cite{liu1}; however, the analogous proof for MIMO systems
% is given in the next section as it is new.

\subsubsection{Noisy COI Decoding ($\sigma^2 > 0$)}

The source is now assumed to have corrupted COI from the destination.
Note that the two main differences between perfect COI decoding and noisy
COI decoding are:
\begin{itemize}
  \item The power allocation variable, $\gamma$, is now a degree of freedom.  This allows the
  source to allocate more or less power to the message signal to adapt to
  conditions of the feedback channel.
  \item The destination can no longer be derived in a simple form as in the
  noiseless feedback case. It is derived from the form of the
  optimal linear estimator of the symbol, $\theta$.
\end{itemize}

\noindent It can be shown that, if $\sigma^2 > 0$, the optimal $\bq$ that maximizes post-processed SNR with the setup in
(\ref{sys}) and (\ref{trans}) is given by
\begin{equation}
\bq =
\frac{\bC^{-1}{\bD}\bg}{\bg^{*}{\bD}^{*}\bC^{-1}{\bD}\bg},
\label{rec}
\end{equation}
\noindent where $\bC = ({\bD}\bF + \bI)({\bD}\bF + \bI)^{*} +
\sigma^2{\bD}\bF\bF^{*}{\bD}^{*}$ is the effective noise covariance
matrix seen at the destination.  Note that this definition of $\bq$ assumes that
all $N$ retransmissions are used as $\bq$ will be a length $N$ vector.  To obtain $\bq_{k}$
where $1 < k < N$, one can simply truncate the vectors (and matrices) in (\ref{rec}) to simply
the first $k$ entries (rows and columns).

With this setup, the post-processed SNR, given the channel coefficients $h[k]$,
can be written as
\begin{equation}
SNR = \rho\left(\bg^{*}{\bD}^{*}\bC^{-1}{\bD}\bg\right).
\label{SNR}
\end{equation}
It is difficult to derive a simple expression for (\ref{SNR}); we instead
formulate bounds on the post-processed SNR.  This is done in the
following lemma for the case of $N = 2$ in the low and high $\rho$ regimes.
\begin{lemma}
Given the linear feedback code described above with blocklength $N = 2$, at
small $\rho$ (i.e., $\rho\ll 1$), the average post-processed $SNR$ can be
bounded by
\begin{equation}
E[SNR_{N=2}]  <  2\rho\left(1 + \sqrt{\gamma}\rho + \gamma
\rho^2\right),
\label{SNRupper}
\end{equation}
and
\begin{equation}
E[SNR_{N=2}]  \underset{\rho \rightarrow 0}{>} 2 \rho \left(1 +
\sqrt{\gamma}\rho - \frac{1 + \sigma^2}{2}\gamma \rho\right).
\label{SNRlower}
\end{equation}
Furthermore, at large $\rho$ (i.e., $\rho\gg1$), the average
post-processed $SNR$ expression behaves as:
\begin{equation}
E[SNR_{N=2}] \underset{\rho \rightarrow \infty}{\longrightarrow} \rho\left( 1 +
\frac{1}{\sigma^2}\right).
\label{SNRhigh}
\end{equation}
\end{lemma}
\begin{proof}
In the case of $N = 2$, the post-processed SNR using (\ref{SNR}) can be
calculated to be
\begin{equation}
SNR_{N=2} = \rho\left( |h[1]|^2 + \frac{\beta_{(\gamma,\sigma^2)}^2[1]
|h[2]|^2 (1+\sqrt{\gamma}\rho |h[1]|^2)^{2}}{1 + \sigma^2
\gamma\rho^2 \beta_{(\gamma,\sigma^2)}^2[1]|h[1]|^2 |h[2]|^2}\right).
\label{SNR2}
\end{equation}

\noindent Using (\ref{betaeq}) which states that $\beta_{(\gamma,\sigma^2)}[k]
< 1$ for any $k$, it is clear that,
\begin{equation*}
SNR_{N=2} < \rho \left(|h[1]|^2 + |h[2]|^2
(1+\sqrt{\gamma}\rho |h[1]|^2)^{2} \right).
\end{equation*}
Now taking expectation on both sides and using the independence of fading blocks in
time and $E[|h[1]|^2] = E[|h[2]|^2] =1, E[|h[1]|^4] = 2$, we immediately
get
\begin{equation}
E[SNR_{N=2}] < 2\rho\left(1 + \sqrt{\gamma}\rho + \gamma
\rho^2\right).
\end{equation}
Using the inequality $(1 + \xi)^{-1} > (1 - \xi)$  valid for any real $\xi$ in
(\ref{SNR2}),
\begin{equation}
SNR_{N=2} >  \rho\left(|h[1]|^2 + \beta_{(\gamma,\sigma^2)}^2[1] |h[2]|^2
(1+\sqrt{\gamma}\rho |h[1]|^2)^{2}\left(1 - \sigma^2
\gamma\rho^2 \beta_{(\gamma,\sigma^2)}^2[1]|h[1]|^2 |h[2]|^2\right)\right).
\label{lb_SNR}
\end{equation}
Taking the conditional expectation with respect to $h[2]$ in (\ref{lb_SNR}), we
get
\begin{eqnarray*}
E\left[SNR_{N=2}\Big{|}h[2]\right] & > & \rho\left(|h[1]|^2 +
\beta_{(\gamma,\sigma^2)}^2[1] \left(1+\sqrt{\gamma}\rho |h[1]|^2\right)^{2}(1 - 2\sigma^2\gamma
\rho^2\beta_{(\gamma,\sigma^2)}^2[1]|h[1]|^2)\right)\\
& = & \rho\left(|h[1]|^2 + \beta_{(\gamma,\sigma^2)}^2[1]
\left(1+\sqrt{\gamma}\rho |h[1]|^2\right)^{2} + O(\rho^2)\right).
\end{eqnarray*}

By the definition of $\beta_{(\gamma,\sigma^2)}^2[1]$ in (\ref{betaeq}) and the
inequality $(1 + \xi)^{-1} > (1 - \xi)$, we have
\begin{eqnarray*}
E\left[SNR_{N=2}\Big{|}h[2]\right] & > & \rho\left(|h[1]|^2 + \left(1 - (1 +
\sigma^2)\gamma \rho |h[1]|^2\right)\left(1+\sqrt{\gamma}\rho |h[1]|^2\right)^{2}
+ O(\rho^2)\right)\\
& = & \rho\left(|h[1]|^2 + 1 + 2\sqrt{\gamma}\rho|h[1]|^2 - (1 +
\sigma^2)\gamma \rho | h[1] |^2 + O(\rho^2)\right).
\end{eqnarray*}
Now taking expectation with respect to the channel realization $h[1]$, we
immediately get
\begin{equation*}
E[SNR_{N=2}]  > 2 \rho \left(1 + \sqrt{\gamma}\rho\left(1 - \frac{1 +
\sigma^2}{2}\sqrt{\gamma}\right) + O(\rho^2) \right).
\end{equation*}
Therefore in the small $\rho$ regime we have,
\begin{equation*}
E[SNR_{N=2}]  \underset{\rho \rightarrow 0}{>} 2 \rho \left(1 +
\sqrt{\gamma}\rho\left(1 - \frac{1 + \sigma^2}{2}\sqrt{\gamma}\right)\right).
\end{equation*}

\noindent Hence for the proposed linear scheme to have better performance than
MRC, we require that $\gamma < \sqrt{\frac{2}{1 + \sigma^2}}$.

\noindent In the case of  large $\rho$, the expression in (\ref{SNR2}) by
approximating the second term can be written as
\begin{eqnarray*}
SNR_{N=2} & \underset{\rho \rightarrow \infty}{\longrightarrow} & \rho\left(
|h[1]|^2 + \frac{\beta_{(\gamma,\sigma^2)}^2[1] |h[2]|^2 (1+\sqrt{\gamma}\rho
|h[1]|^2)^{2}}{\sigma^2 \gamma\rho^2 \beta_{(\gamma,\sigma^2)}^2[1]|h[1]|^2
|h[2]|^2}\right)\\ & = & \rho\left( |h[1]|^2 + \frac{(1+\sqrt{\gamma}\rho
|h[1]|^2)^{2}}{\sigma^2 \gamma\rho^2 |h[1]|^2}\right)\\
& \underset{\rho \rightarrow \infty}{\longrightarrow} & \rho\left( |h[1]|^2 +
\frac{\gamma\rho^2 |h[1]|^4}{\sigma^2 \gamma\rho^2 |h[1]|^2}\right)\\
& = & \rho\left( |h[1]|^2 + \frac{1}{\sigma^2}|h[1]|^2\right).
\end{eqnarray*}
Now, taking expectation we get,
\begin{equation*}
E[SNR_{N=2}] \underset{\rho \rightarrow \infty}{\longrightarrow} \rho\left( 1 +
\frac{1}{\sigma^2}\right).
\end{equation*}
\noindent Note that when $\sigma^2 \approx 0$, the above scheme yields
significant benefits.
\end{proof}

\subsection{Power Allocation}
In this subsection, we investigate the power allocation parameter $\gamma$ seen
in the scheme for noisy COI feedback.  As stated before, it can be roughly thought
of as a measure of the amount of feedback side-information being used in the
retransmission. Optimally choosing the value of $\gamma$ to maximize the
post-processed SNR in (\ref{SNR}) is clearly a non-causal
problem. Therefore instead we define $\gamma_0^{(\rm fading)}$ to be the one
that maximizes the post-processed SNR over the ensemble average of all channel
realizations, i.e.,
\begin{equation}
\gamma_{0}^{(\rm fading)} = \max_{\gamma\in[0,1]}
E\left[\rho(\bg^{*}{\bD}^{*}\bC^{-1}{\bD}\bg)\right].
\end{equation}

\begin{figure}
\centering
\includegraphics[scale = 0.5]{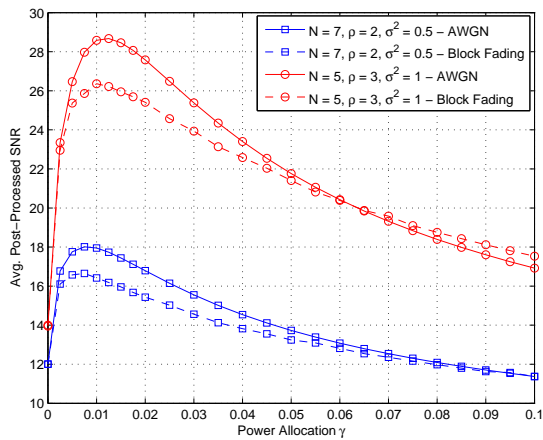}\\
\caption{Average post-processed SNR versus choice of power allocation,
$\gamma$, for AWGN and block fading.}\label{gammafig}
\end{figure}

The difficulty of analytically calculating the above quantity stems from the
post-processed SNR having non-linear dependencies on the fading coefficients
$h{[1]},\ldots,h{[N]}$. However, it turns out that the optimal $\gamma$ in the
i.i.d. Rayleigh block fading case ($\gamma_{0}^{(\rm{fading})}$) is very close
to the optimal $\gamma$ in the AWGN case ($\gamma_{0}^{(\rm{AWGN})}$) as derived
in \cite{ZaDa11}.  This is displayed in Fig. \ref{gammafig}. In Fig.
\ref{gammafig}, we see that the peaks of both performance curves for block
fading (averaged over 15,000 trials) and AWGN noise are quite close together.
This is quite beneficial as it is very easy to numerically find the value of
$\gamma$ that maximizes the post-processed SNR in the AWGN case, whereas it
proves to be much more difficult in the presence of block fading. Because of the
proximity of $\gamma_{0}^{(\rm{AWGN})}$ and $\gamma_{0}^{(\rm{fading})}$, we
assume that the value of $\gamma$ that maximizes the average post-processed SNR,
$\gamma_{0} = \gamma_{0}^{(\rm{AWGN})} \approx \gamma_{0}^{(\rm{fading})}$.  The
value of $\gamma_{0}$ does, however, change with the blocklength $N$.
Furthermore, as the number of transmissions is not necessarily known ahead of
time, it is intuitive to not choose $\gamma$ as a function of blocklength.
Alternatively, we can fix $\gamma$ based on a reasonable number of packet
retransmissions---this is discussed in the following example.

\subsection*{Example 1}
To illustrate the performance of the linear feedback scheme, we now provide some
simulations.  In this first plot (Fig. \ref{snrfig}), the post-processed SNR of
the scheme is plotted in contrast to MRC.  MRC is analogous to using our scheme
but setting $\gamma = 0$.  In other words, the source simply repeats the
packet at each retransmission.  Then, retransmissions are combined using a
linear receiver similar to the one in (\ref{rec}).  The simulations were run with an
average transmit power of $\rho = 3$ and for both noiseless COI feedback and
varying levels of noisy COI feedback. As can be seen, the linear feedback
outperforms MRC with a gap that increases with decreasing feedback noise.

\begin{figure}
\centering
\subfloat[Average Post-processed SNR]{\label{snrfig}\includegraphics[scale = 0.5]{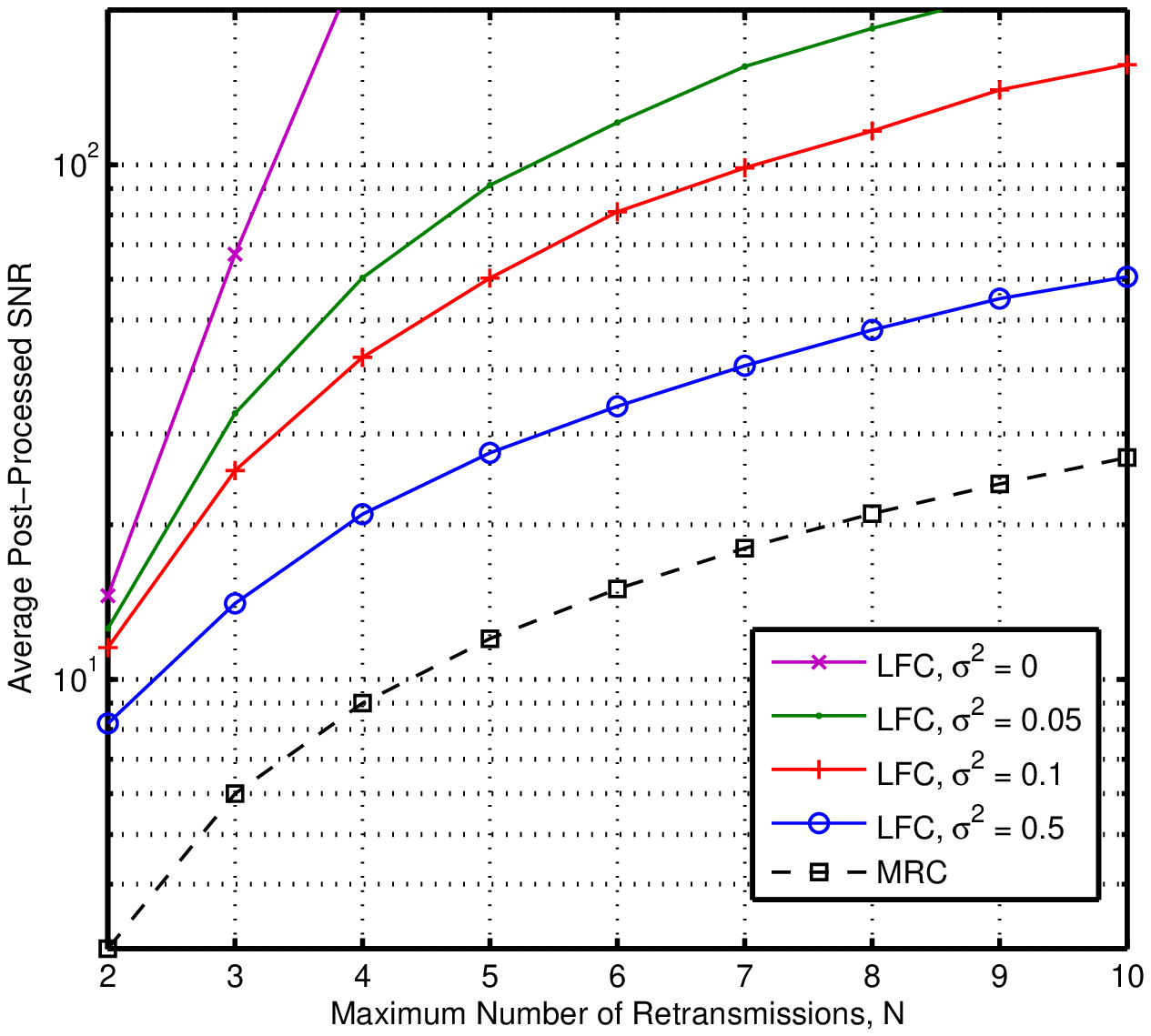}}
\subfloat[Power Allocation]{\label{gammasnr}\includegraphics[scale = 0.5]{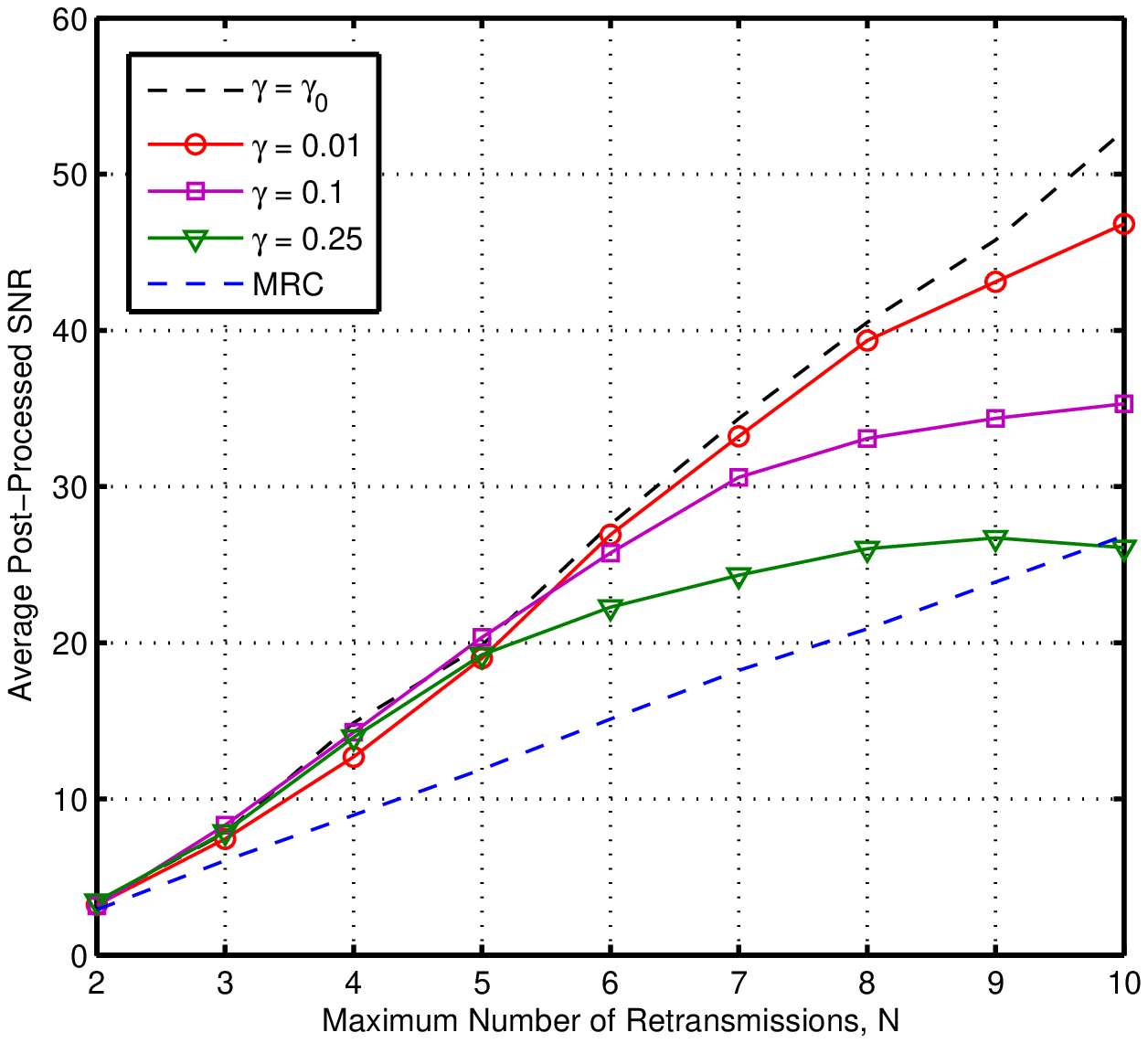}}
\label{gamma_plots}
\caption{(a) Average post-processed SNR performance of linear feedback combining (LFC) versus
maximal ratio combining. (b) Average post-processed SNR performance of linear feedback combining using
different values of $\gamma$.}
\end{figure}

As mentioned above, $\gamma_{0}$ changes with blocklength, $N$, and therefore
should be chosen appropriately.  However, in hybrid-ARQ, the required number of
retransmissions is often not known ahead of time.  Despite this fact, not having
this knowledge a priori provides very little penalty to performance. If the
number of retransmissions is not assumed to be predetermined, $\gamma$ can be
approximately chosen using the feedback noise variance $\sigma^2$ and the
average transmit power $\rho$.  The next figure, Fig. \ref{gammasnr},
illustrates the effect of fixing $\gamma$. As is illustrated, fixing $\gamma$
with respect to blocklength yields little performance degradation as long as
$\gamma$ is chosen appropriately.  The average post-processed SNR for $\gamma =
0.01$ performs very close to the scheme when using $\gamma_{0}$ from Fig.
\ref{gammafig}.  Note that Fig. \ref{gammasnr} has been plotted on a linear
scale to help display the comparison.

\section{Multiple Antenna Scenarios}\label{sec:mas}
In this section, we show how the feedback scheme for SISO systems
can be implemented in both MISO and MIMO systems with current CSI
at the source. In addition, an extension of the scheme is given for MIMO
systems with perfect COI and only outdated CSI at the source. However,
first we look at a MISO system with current quantized CSI along
with perfect COI available at the source.

\subsection{MISO with Current, Quantized Channel State Information at the
Source} Consider a MISO discrete-time system (Fig. \ref{miso_main}) with $M_t$ transmit antennas and only one
receive antenna, where the received packet, $\by[k] \in \mathbb{C}^{1\times L}$ is given by
\begin{equation}\label{miso1}
\by[k] = \bh^T [k] \bX[k] + \bz[k], \quad k = 1, \ldots N,
\end{equation}
where $\bh[k] \in \mathbb{C}^{M_t \times 1}$ is the channel gain vector, $\bX[k] \in
\mathbb{C}^{M_t \times L}$ is the transmitted packet matrix where the columns
correspond to channel uses and the rows correspond to antennas, and $\bz[k] \in
\mathbb{C}^{1\times L}$ is additive noise during the $k^{th}$ retransmission
with distribution $\cC\cN(\mathbf{0},\bI)$.  Furthermore, the power constraint
at the source is given as $E [\mathrm{tr}(\bX^*[k]\bX[k])] \leq
L\rho$, and it is assumed that there is perfect CSI at the destination.
However, the source no longer has access to perfect CSI. The
destination only feeds back the beamforming vector to be used for
current packet retransmission. The previous channel quality information ($\bh^T[k-1]\bw[k-1]$) along with the
unquantized channel output is also fed back to the source.

\begin{figure*}
\centering
\includegraphics[scale = 0.5]{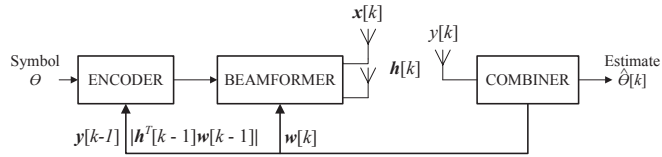}
\caption{System level block diagram for MISO system with quantized CSI
and COI feedback.  The feedback information can be described as COI ($\by[k-1]$), channel quality information ($\bh^T[k-1]\bw[k-1]$), and
the beamforming vector ($\bw[k]$).}\label{miso_main}
\end{figure*}

The transmitted packet matrix $\bX[k]$ is now generated as an outer product by
\begin{equation}\label{miso2}
\bX[k] = \bw[k]\widetilde{\bx}[k],
\end{equation}
where $\bw[k] \in \mathbb{C}^{M_{t} \times 1}$ denotes the unit norm beamforming vector to be used during
retransmission $k$ and $\widetilde{\bx}[k]\in \mathbb{C}^{1 \times L}$ is the
effective SISO signal during retransmission number $k$. The power constraint on
$\bX[k]$ now is equivalent to
\begin{eqnarray}
\nonumber E[\mathrm{tr}(\bX^*[k]\bX[k])] & = & \mathrm{tr}\left(E[\widetilde{\bx}^*[k]\bw^*[k]\bw[k]
\widetilde{\bx}[k]]\right)\\
%\nonumber & = & \mE[\widetilde{\bx}^*[k]\widetilde{\bx}[k]]\\
\nonumber & = & E[|\widetilde{\bx}[k]|^2]\\
& \leq & L\rho.
\end{eqnarray}
At this point, it is again assumed that $L = 1$ for simplicity which reduces
$\by[k], \widetilde{\bx}[k]$, and $\bz[k]$ to scalars $y[k], \widetilde{x}[k]$,
and $z[k]$. We now follow the standard model for limited feedback beamforming by
constraining the design of beamforming vector $\bw[k]$ for packet transmission
$k$ to a codebook $\cF[k]$ containing $2^B$ unit vectors \cite{love2}. We denote the codebook
$\cF[k]$ as
\begin{equation}
\cF[k] = \left\{ \bff _1[k], \ldots, \bff _{2^B}[k]\right\}, \quad \lVert
\bff _j[k] \rVert = 1, 1 \leq j \leq 2^B.
\end{equation}
We can use any scheme available in literature to generate the unit beamforming
vectors including random vector quantization (RVQ)~\cite{honig},\thinspace
\cite{honig2} and Grassmannian line packing~\cite{love3,Mukka}. This codebook is
accessible to both the source and destination simultaneously. For RVQ, there
must be a random seed that is made available to both the source and destination
before the communication starts.

The destination decides on the beamforming vector that the source
uses during the $k^{th}$ retransmission by solving the following
channel quality maximization problem
\begin{equation}
\bw[k] = \argmax_{\bff _j[k] \in \cF[k]}\left|\bh^T[k] \bff_j[k]\right|^2.
\end{equation}
Effectively, the destination chooses the unit vector $\bw[k]$ in the
codebook $\cF [k]$ along which the channel vector $\bh[k]$ has the largest projection. The
information about $\bw[k]$ is conveyed back to the source in just $B$
bits. The limited feedback capacity~($C_{\rm{LF}}$) for a given codebook
design $\{\cF[k]\}_{k = 1}^{\infty}$ can be expressed by
\begin{equation}
C_{\rm{LF}} = E\left[\max_{\bff _j[k] \in \cF[k]}\log_2(1 + \rho |\bh^T[k]
\bff_j[k]|^2)\right].
\end{equation}
Using the monotonicity of the logarithmic function, $C_{\rm{LF}}$ can be
simplified to
\begin{eqnarray}
\nonumber C_{\rm{LF}} & = & E\left[\log_2(1 + \rho \max_{\bff _j[k] \in
\cF[k]}|\bh^T[k] \bff_j[k]|^2)\right]\\
& = & E\left[\log_2(1 + \rho |\bh^T[k]
\bw[k]|^2)\right].
\end{eqnarray}
As the number of feedback bits $B$ approach infinity, $C_{\rm{LF}} \rightarrow
C_{\rm{MISO}}$, where $C_{\rm{MISO}} = E\left[\log_2(1 + \rho\lVert\bh
\rVert^2)\right]$. This is because limited CSI feedback becomes perfect CSI
feedback for any codebook design with an infinite number of feedback bits. With
the selection of beamforming vector $\bw[k]$ as described above and packet
length $L = 1$, the received signal $y[k]$ is given as
\begin{equation*}
y[k] = \bh^T [k] \bw[k] \widetilde{x}[k] + z[k], \quad k = 1, \ldots N.
\end{equation*}
Pre-multiplying the received signal $\by[k]$ by $e^{-j\angle \bh^T [k] \bw[k]}$,
we obtain
\begin{equation*}
\widetilde{y}[k] = \left|\bh^T [k] \bw[k]\right|\widetilde{x}[k] +
\widetilde{z}[k], \quad k = 1, \ldots N,
\end{equation*}
where $\widetilde{y}[k] = y[k]e^{-j\angle \bh^T [k] \bw[k]}$ and $\widetilde{z}[k]$ is distributed as $\cC\cN(0,1)$. If we let $\widetilde{\lambda}[k] = |\bh^T
[k] \bw[k]|$, we get the overall system in (\ref{miso1}) as
\begin{equation}
\widetilde{y}[k] = \widetilde{\lambda}[k]\widetilde{x}[k] + \widetilde{z}[k], \quad k =
1, \ldots N,\label{tildes}
\end{equation}
\noindent Finally, gathering all packet retransmissions together as in
(\ref{sys}), we can rewrite (\ref{tildes}) as
\begin{equation}
\widetilde{\by}= \widetilde{\bLambda}\widetilde{\bx} + \widetilde{\bz},
\end{equation}
\noindent where $\widetilde{\bLambda} =
\mathrm{diag}(\widetilde{\lambda}[1],\widetilde{\lambda}[2],\ldots,\widetilde{\lambda}[N])$.
With this formulation, the MISO system is equivalent to the SISO system in
(\ref{sys}); therefore, the SISO scheme can be implemented by replacing the
role of $h[k]$ with $\widetilde{\lambda}[k]$ (see Fig. \ref{miso_main}).

It can be proven in a similar way to Lemma 1 that the MISO scheme meets the average power constraint.  Also, it can be
shown that if the feedback channel is perfect, the MISO scheme achieves the capacity of the channel and obtains a
doubly exponential decay in error probability.  However, to avoid redundancy,
this proof is only given for the MIMO case in the next section (Lemma 4). The
effects of using different vector quantization techniques and the overall
performance of the MISO scheme are now presented in an example.

\begin{figure}
	\centering
		\includegraphics[scale = 0.5]{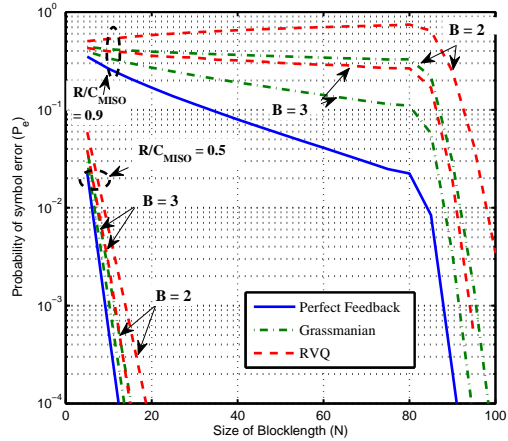}
	\caption{Variation of probability of error with the number of retransmissions
	for i.i.d. Rayleigh fading at $\rho = 0$\thinspace dB and $M_t = 2$. The
	performance of RVQ, Grassmanian line packing, and perfect feedback are compared
	for $R/C_{\rm{MISO}} = 0.5$ and $R/C_{\rm{MISO}} = 0.9$ with $B = 2$ and $B = 3$.}
	\label{fig3_miso}
\end{figure}

\subsection*{Example 2}
To illustrate the potential of our scheme, consider a MISO system
communicating over an i.i.d. Rayleigh block fading channel with each entry of
$\bh[k]$ distributed as $\cC \cN({0, 1})$. In this example, the COI feedback
is assumed to be noiseless (i.e., $\sigma^2 = 0$).  Using a limited
CSI feedback framework, Fig. \ref{fig3_miso} plots the packet probability of
error curves against the number of retransmissions for two different normalized
rates of $0.5$ and $0.9$ where normalized rate is the ratio of the rate of
transmission to the channel capacity ($R/C_{\rm{MISO}}$). The plots are for
$\rho = 0$\thinspace dB with a two-antenna source averaged over $10^6$
i.i.d. fading realizations. The doubly exponential decay of the curves are
clearly visible for all the feedback schemes: perfect CSI feedback and
quantized CSI feedback -- RVQ and Grassmanian line packing. Even with
quantized CSI feedback and moderate normalized rate of $0.5$, only a few
retransmissions are required to achieve a very low packet error rate of $1\%$
for both RVQ and Grassmanian line packing.

\subsection{MIMO with Current Channel State Information at the Source}

\begin{figure*}
\centering
\includegraphics[scale = 0.49]{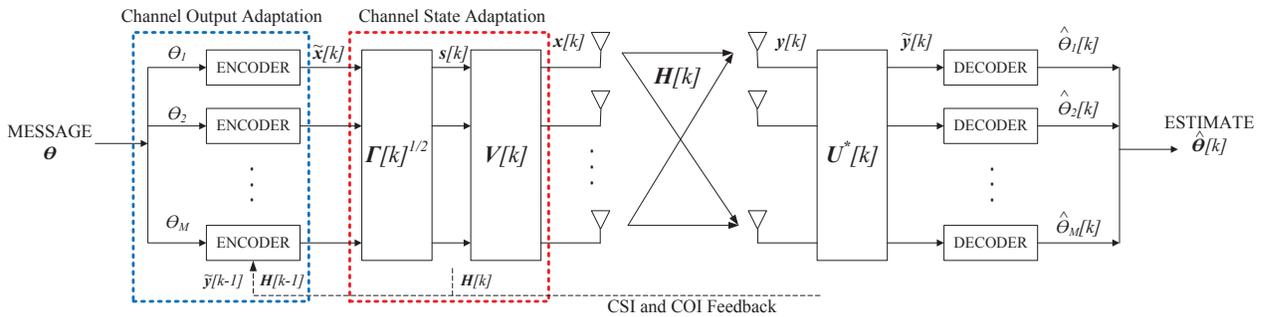}
\caption{System level block diagram when channel state is
known perfectly at both the source and destination.}\label{mimo_sep}
\end{figure*}

Consider now a MIMO packet retransmission system (Fig. \ref{mimo_sep}) with
$M_{t}$ transmit antennas and $M_{r}$ receive antennas where the number of
spatial channels available is $M = \min(M_{r},M_{t})$. The received matrix,
$\bY[k] \in \mathbb{C}^{M_{r}\times L}$, is given by
\begin{equation}
\bY[k] = \bH[k]\bX[k] + \bZ[k],\label{mimosys}
\end{equation}
\noindent where $\bX[k] \in \mathbb{C}^{M_{t}\times L}$ is, as in MISO, the
transmit packet matrix, $\bH[k] \in \mathbb{C}^{M_{r}\times M_{t}}$ is the block
Rayleigh fading channel matrix whose entries are i.i.d. zero-mean complex
Gaussian random variables with unit variance, and $\bZ[k] \in
\mathbb{C}^{M_{r}\times L}$ is an additive noise matrix with i.i.d. zero-mean
complex Gaussian entries with unit variance. Note that due to the availability
of multiple spatial channels, the total packet length has increased to contain
$ML$ symbols with $M$ symbols transmitted over each channel use. Again, for the
sake of simplicity, we assume that $L = 1$ which reduces $\bY[k], \bX[k],$ and
$\bZ[k]$ to column vectors $\vec{\by}[k], \vec{\bx}[k],$ and $\vec{\bz}[k]$.
When the current block fading matrix is known both at the source and
destination, we can effectively diagonalize the channel. Let
\begin{equation}\label{svd1}
\bH[k] = \bU[k] \mathbf{\Lambda}[k] \bV^{*}[k],
\end{equation}
be a compact singular value decomposition~(SVD) of the channel matrix
$\bH[k]$, where $\bU[k] \in \mC^{M_r \times M}, \mathbf{\Lambda}[k] \in \mC^{M
\times M},$ and $\bV[k] \in \mC^{M_t \times M},$ with
\begin{equation}\label{svd2}
\mathbf{\Lambda}[k] = \mathrm{diag}\left(
\lambda_1[k], \ldots , \lambda_M[k] \right),\quad \lambda_1[k] \geq \lambda_2[k]
\ldots \geq \lambda_M[k] \geq 0,
\end{equation}
\begin{equation}\label{svd3}
\bU^* [k] \bU[k] =  \bV^* [k]\bV[k] = \bI.
\end{equation}
We can design the source vector $\vec{\bx}[k]$ as
\begin{equation}\label{x1}
\vec{\bx}[k] = \bV[k]\bs[k],
\end{equation}
where $\bs[k] \in \mathbb{C}^{M \times 1}$ with $\bV[k]$ defined
by~(\ref{svd1}) and~(\ref{svd3}). Also pre-multiplying the received vector by
$\bU^*[k]$, we obtain the effective system described by~(\ref{mimosys}) as
\begin{eqnarray}
\nonumber \bU^{*}[k]\vec{\by}[k] & = & \bU^{*}[k] \bH[k] \bV[k] \bs[k] +
\bU^{*}[k]\vec{\bz}[k]\\
\widetilde{\by}[k] & = & \mathbf{\Lambda}[k]\bs[k] + \widetilde{\bz}[k],
\label{effmimosys}
\end{eqnarray}
where $\widetilde{\by}[k] \in \mC^{M \times 1}$ and $\widetilde{\bz}[k] \in
\mC^{M \times 1}$. The effective noise $\widetilde{\bz}[k]$ is distributed as
$\mathcal{CN}({\bf 0},\bI)$ due to the rotational invariance of complex i.i.d. Gaussian vectors.
Due to the a priori knowledge of the channel at the source, spatial
waterfilling can be performed across the $M$ parallel spatial channels for
each packet transmitted. The entries of the waterfilling matrix
$\mathbf{\Xi}[k] = \textrm{diag}\left(
\xi_1[k], \ldots , \xi_M[k] \right)$ are defined as
\begin{equation}\label{wfill2}
\xi_i[k] = \max\left(0, \frac{1}{\xi_0[k]} -
\frac{1}{{\lambda^2_i}[k]}\right), \quad 1 \leq i \leq M.
\end{equation}
The value of the constant $\xi_0[k]$ is the water-filling level chosen
to satisfy the power constraint
\begin{equation}\label{wfill3}
\sum_{i = 1}^{M}\xi_i[k] = 1.
\end{equation}
Furthermore, the capacity $C_{\rm{TR}}$ of a MIMO channel with the fading matrix
known both at the source and destination can be written as
\begin{equation}
C_{\rm{TR}} = \sum_{i = 1}^{M}E\left[\log_2(1 + \rho
\xi_i{\lambda}_i^2)\right], %= \sum_{i = 1}^{M}C_i.
\end{equation}
where we have dropped the retransmission index $k$ due to the i.i.d. nature of
the block fading matrix. With current CSI at the source and
destination, the overall channel capacity of the MIMO channel can be
expressed as a sum of $M$ parallel non-interfering SISO spatial channels
each with capacity $C_i$ where $C_i = E\left[\log_2(1 + \rho
\xi_i{\lambda}_i^2)\right], \quad 1 \leq i \leq M.$

With the aid of the waterfilling matrix defined in (\ref{wfill2}),
(\ref{effmimosys}) can now be written as
\begin{eqnarray}
\nonumber \widetilde{\by}[k] & =
& {\mathbf{\Lambda}}[k]\mathbf{\Xi}[k]^{1/2}\widetilde{{\bx}}[k] +
\widetilde{\bz}[k],
\end{eqnarray}
where $\bs[k] = \mathbf{\Xi}[k]^{1/2}\widetilde{{\bx}}[k]$. Note that the
spatial waterfilling~(or power adaptation) does not make use of the COI fed back
to the source at all. Letting $\widetilde{\mathbf{\Lambda}}[k] =
{\mathbf{\Lambda}}[k]\mathbf{\Xi}[k]^{1/2}$, the overall system can be
represented in matrix form as
\begin{equation}\label{effsys1}
\widetilde{\by}[k] = \widetilde{\mathbf{\Lambda}}[k]\widetilde{{\bx}}[k] +
\widetilde{\bz}[k].
\end{equation}
We next transmit $M$ symbols over $M$ parallel spatial channels by exploiting
the COI and previous CSI available at the source using a maximum of $N$
transmissions.  In other words, with (\ref{effsys1}), we can implement $M$
parallel instances of the COI feedback SISO scheme---one for each spatial
channel. Similar to the MISO case, we replace the role of $h[k]$ with
$\widetilde{\lambda}_{i}[k]$ for the $i^{th}$ spatial channel.

It is quite possible that each of the source constellations~$\varTheta_i[N]$ has
a different number of constellation points; note that $\varTheta_i[N]$ denotes
the source constellation used for the $i^{th}$ spatial channel. The number of
equally likely constellation points chosen for each channel depends on the
spatial capacity $C_i$ of the subchannel. Therefore, the number of
constellations points must be less than $2^{NC_i}$.

The overall schematic of the proposed scheme, shown in
Fig. \ref{mimo_sep}, clearly demonstrates the independent constellation
mapping of each of the $M$ symbols of packet $\btheta$
along with the separation of the channel output adaptation from current channel
state adaptation. Furthermore, it can be shown that, if the feedback channel is perfect, any rate less than capacity
can be achieved by the above scheme at doubly exponential rate.

\begin{lemma}\label{lem4}
If $\sigma^2 = 0$, the proposed scheme achieves any rate $R < C_{\rm{TR}}$.
Viewing the rate $R$ as a sum of $M$ spatial channel
rates, $R = \sum_{i = 1}^{M}R_i,$ the coding scheme can achieve any rate $R_i < C_i$ for
the $i^{th}$ spatial channel. Furthermore the probability of error~($P_e$) for
the packet decays doubly exponentially as the function of the number
of transmissions $N$. In other words, for sufficiently large $N$,
\begin{equation*}
P_e \leq \beta_1\exp\left(-2^{(N\beta_2 + \beta_3)}\right),
\end{equation*}
where $\beta_1$ and $\beta_2$ are positive constants, while $\beta_3$ is a real
constant for a given rate $R$.
\end{lemma}
\begin{proof}
See Appendix.
\end{proof}

\subsection{MIMO with Outdated Channel State Information at the Source}
In the case when there are multiple antennas at both the source and
destination and the source has access to only outdated CSI, a direct extension of
the SISO scheme for perfect COI (Lemma \ref{lem2}) can be made. Using the same
system setup as in (\ref{mimosys}), if $L = 1$, we can write the feedback scheme recursively as
\begin{eqnarray}
\label{recenc}
\vec{\bx}[k + 1] & = & \left(\bI +
\rho\bH^*[k]\bH[k]\right)^{-1/2}\left(\vec{\bx}[k] -
\rho\bH^*[k]\vec{\bz}[k]\right)\\ \label{recdec}
\widehat{{\textrm{{\boldmath{$\theta$}}}}}[k] & = & \left(\bI -
\mathbf{\Phi}[k]\mathbf{\Phi}^*[k] \right)\textrm{{\boldmath{$\theta$}}} + \rho
\mathbf{\Phi}[k]\mathbf{\Phi}^*[k]\sum_{m = 1}^{k}(\mathbf{\Phi}^{-1}[m -
1])^{*}\bH^*[m]\vec{\bz}[m],
\end{eqnarray}
\noindent where
\begin{equation}
\mathbf{\Phi}[k] =  \left(\bI + M\rho\bH^*[1]\bH[1] \right)^{-1/2}\cdots
\left(\bI + M\rho\bH^*[k]\bH[k] \right)^{-1/2}.
\end{equation}

Unfortunately, even if the feedback channel is perfect but only outdated
CSI is available at the source, it is difficult to prove a result
similar to Lemma 4.  Although it can be shown for some positive rates a doubly exponential
decay of probability of error is achievable, it has not been proven for all
rates below capacity. We now broaden our focus back to the view of the whole
hybrid-ARQ scheme in the next section.

\section{The Hybrid-ARQ Scheme and Variations}\label{sec:harq}

Rather than focusing on the packet estimate, $\widehat{\btheta}[k]$, we now
consider the overall hybrid-ARQ scheme including the FEC. For the FEC, we assume the use of a systematic turbo
code; although, any systematic block code can work.  Because of this choice, we
will perform only soft decoding at the output of the packet combining step.
This means for each symbol in the estimated packet, $\widehat{\theta}_{i}[k]$,
we will form a set of log-likelihood ratios as
\begin{equation}
LLR_{i} = \left\{\log\left[\frac{p(\theta_{i} = \psi_{j}|\hat{\theta}_{i}[k])}{\sum_{\ell \neq j}p(\theta_{i} = \psi_{\ell}|\hat{\theta}_{i}[k])}\right]: j = 1,2,\ldots,\left|\Theta[N]\right|\right\},
\end{equation}
\noindent where $\psi_{j}$ for $j = 1,2,\ldots,\left|\Theta[N]\right|$ are the
points of the constellation $\Theta[N]$ utilized for modulation.  Upon
calculating these sets, they are then passed to the turbo decoder for decoding.
The specific turbo code implemented is a rate $L_{\textrm{info}}/L_{\textrm{coded}} = 1/3$ turbo code defined in the UMTS
standard; more details are given in the Simulations section.

Now, we introduce different configurations of the overall scheme that might help adapt
to different circumstances (e.g., feedback link rate, transmit/receive duration,
etc.).  To do so, we look at varying the amount of COI feedback sent to the
source; this is also done to illustrate the trade-off between performance (e.g.,
throughput, FER, etc.) and the amount of information fed back. Note that the
case of CSI-only feedback has already been explored in the literature; see for
example \cite{SuMa06, SaDi06, SuDi06, SuDi07, KiKa09}. Therefore the emphasis
here is in varying the amount of COI feedback.

The most straightforward of the possible COI feedback configurations is one
where the destination simply feeds back everything it receives without
discrimination. This utilizes a noiseless/noisy version of the full received packet for feedback
information; hence, we will refer to this method as \emph{full packet feedback}
(FPF).  Alternatively, one can alter FPF by implementing a well-known concept in
hybrid-ARQ with feedback \cite{Matsu}; instead of feeding back all the symbols
of the received packet, we can instead feed back only the $T$ ``least
reliable'' symbols with their indices.  The measure of
``reliability'' can be based off metrics such as the log-likelihood ratio (LLR)
or the logarithm of the a posteriori probabilities (log-APP) \cite{Shea}.  Since only some of the symbols
in the packet are fed back, we will refer to this scheme as \emph{partial packet
feedback} (PPF).

\subsection{Full Packet Feedback~(FPF)}
In FPF, we propose a hybrid-ARQ scheme where the source is assumed to have
access to a noiseless/noisy version of the last received packet. The performance
of the perfect COI feedback scenario is used to demonstrate the maximum
possible gains that can be achieved with the addition of channel output feedback. To help explicitly
show the feedback information available, we now introduce $\br[k]$ as the
channel output feedback side-information available at the source at
packet transmission $k$. In FPF,
\begin{equation}
\br[k] = \by[k-1] + \bn[k-1],
\end{equation}
\noindent where, in this case, $\br[k] \in \mathbb{C}^{1\times L}$.

As mentioned before, the first transmission of packet
$\textrm{{\boldmath{$\theta$}}}$ is assumed to be a codeword of a FEC code.  If
a NACK is received at the source, each subsequent packet is encoded symbol-wise
by the linear feedback code described in Section III.  This is used to refine
the destination's estimate of each symbol in the original packet.  To display
the performance of the scheme, we look at comparing the normalized throughput of
this scheme with the turbo-coded hybrid-ARQ used in \cite{HSDPA}. This standard
uses a rate-compatible punctured turbo code to encode the packet. Specifically,
it uses a rate 1/3 UMTS turbo code \cite{UMTS} and then punctures it for use in
hybrid-ARQ. If sending one packet and $M$ spatial channels are available for the
MIMO setting, the assignment of $M$ symbols for $M$ spatial channels is done
arbitrarily. Note that it is plausible that using dynamic adaptive modulation
for each of the spatial channels or coordinating multiple
retransmissions\cite{SuMa06} might result in improvement in throughput. However,
we do not consider this here, but we point out that in most of the cases our
proposed scheme can be combined with the innovations on using CSI more
efficiently.

\subsection{Partial Packet Feedback~(PPF)}

For sake of practicality, it is desirable to minimize the amount of COI feedback
information needed to be sent back to the source.  As a step towards this, we
now look at the effects of limiting the size of the COI feedback packet. We try
to utilize the limited feedback channel in the most useful way by feeding back
not the complete packet but only relatively few of the symbols in the received
packet. As mentioned above, in the partial packet mode, the choice of COI
feedback information is based on the relative reliability of soft decoded bits.
This addition to the scheme is motivated by the technique used in \cite{Shea}
where it was shown that focusing on the least reliable information bits can
greatly improve the performance of turbo-coded hybrid-ARQ. The selection process
to construct the feedback packet, $\br[k]$, is performed at the destination
using the following method. The received packet $\by[k]$ (or $\bY[k]$ as in
(\ref{mimosys})) is combined with the $k - 1$ previous received packets using
MRC in the case of Chase combining or as described above if linear COI feedback
coding is employed. After combining, the destination now has an estimate of the
desired packet, $\widehat{\btheta}[k]$. This packet estimate is now passed on to
the turbo decoder, and its corresponding output is a set of LLRs for each
original information bit. For notation, we refer to the LLR produced by the
turbo decoder for the $i^{th}$ information bit, $m_i$, as $\ell_i$, which can be
mathematically written as
\begin{equation}
\ell_i = \log\left[ \frac{p(m_i = 0| \by[1])}{p(m_i = 1| \by[1])}\right], \quad
1 \leq i \leq L_{\rm info} .
\end{equation}

The least reliable bits are chosen as the $T$ bits whose LLR values
have the smallest magnitude, i.e., $|\ell_i| < \ell_{\rm th},$ where
$\ell_{\rm th}$ is chosen appropriately. The magnitude of $\ell_i$
close to zero indicates that the bit is almost equally likely to be either a 1 or a
0. Then, the set of symbols whose realizations are to be fed back is
\begin{equation}
I_{\rm{sym}} = \left\{\theta_k : m_i \in \theta_k \textrm{ and } |\ell_i| <
\ell_{\rm th}, \quad 1 \leq i \leq L_{\rm{info}}, \quad 1 \leq k \leq L\right\},
\end{equation}
\noindent meaning the symbol is chosen to be fed back if it contains one of
the least reliable information bits. Let $|I_{\rm sym}| = T_{\rm sym}$. With
this technique, we can then write the feedback packet, $\br_T[k]$, as
\begin{equation}
\br_T[k] = \by_{T}[k-1] + \bn_T[k-1],
\end{equation}
\noindent where
\begin{equation}
\by_{T}[k-1] = \left\{y_{i}[k-1]: \theta_i \in I_{\rm{sym}}\right\}.
\end{equation}
\noindent  Since $T_{\rm sym}$ channel outputs are being fed back, $\bn_T[k]$ is
of length $T_{\rm sym}$. Note that the selection process is straightforward in
this case because we assume the use of a systematic turbo code. It is also
important to note that the $T_{\rm sym}$ symbols are chosen only once (after the first
transmission). This process can be done after each retransmission but would
require more feedback resources. Finally, if it is assumed that the number of
channel uses per packet retransmission are constant, one can fill in~the
remaining $L - T_{\rm sym}$ channel uses in numerous ways. One particular way is what we will refer to as
partial packet feedback with partial Chase combining (PPF-PC). In this mode, on
the forward transmission, the new $T_{\rm sym}$ symbols generated for the
$T_{\rm sym}$ least reliable symbols based on our linear coding scheme are sent in conjunction with
the repetition of the remaining $(L - T_{\rm sym})$ other symbols used for
Chase combining.

\section{Simulations}\label{sec:sim}
In this section we present numerical simulations to demonstrate the improvements
possible with inclusion of our proposed linear COI feedback coding in hybrid-ARQ
schemes. We assume that the channel is i.i.d. Rayleigh block fading. We
limit the number of retransmissions to a maximum of four (i.e., $N = 4$). All
the throughput calculations are done by averaging over $10^3$ new packet
transmissions.

The metric defined for calculating normalized throughput is given as: \[\tau =
\frac{1}{E[B]},\] where $1 \leq B \leq N$ is the number of transmissions needed
for successful decoding of a packet $\btheta$.  This can be equivalently thought
of as a packet success rate or the inverse of the average number of packets
needed for successful transmission. If the number of retransmissions reaches the maximum number
before successful decoding, the throughput contribution is zero. Note that this
metric is meaningful only when comparing constant-length packet schemes. Also
the above throughput definition implies that as $\rho \rightarrow \infty, \tau
\rightarrow 1$ for all the protocols; including Chase and our proposed scheme.

\begin{figure}
\centering
\subfloat[MIMO Channel]{\label{fpf_plot}\includegraphics[scale = 0.55]{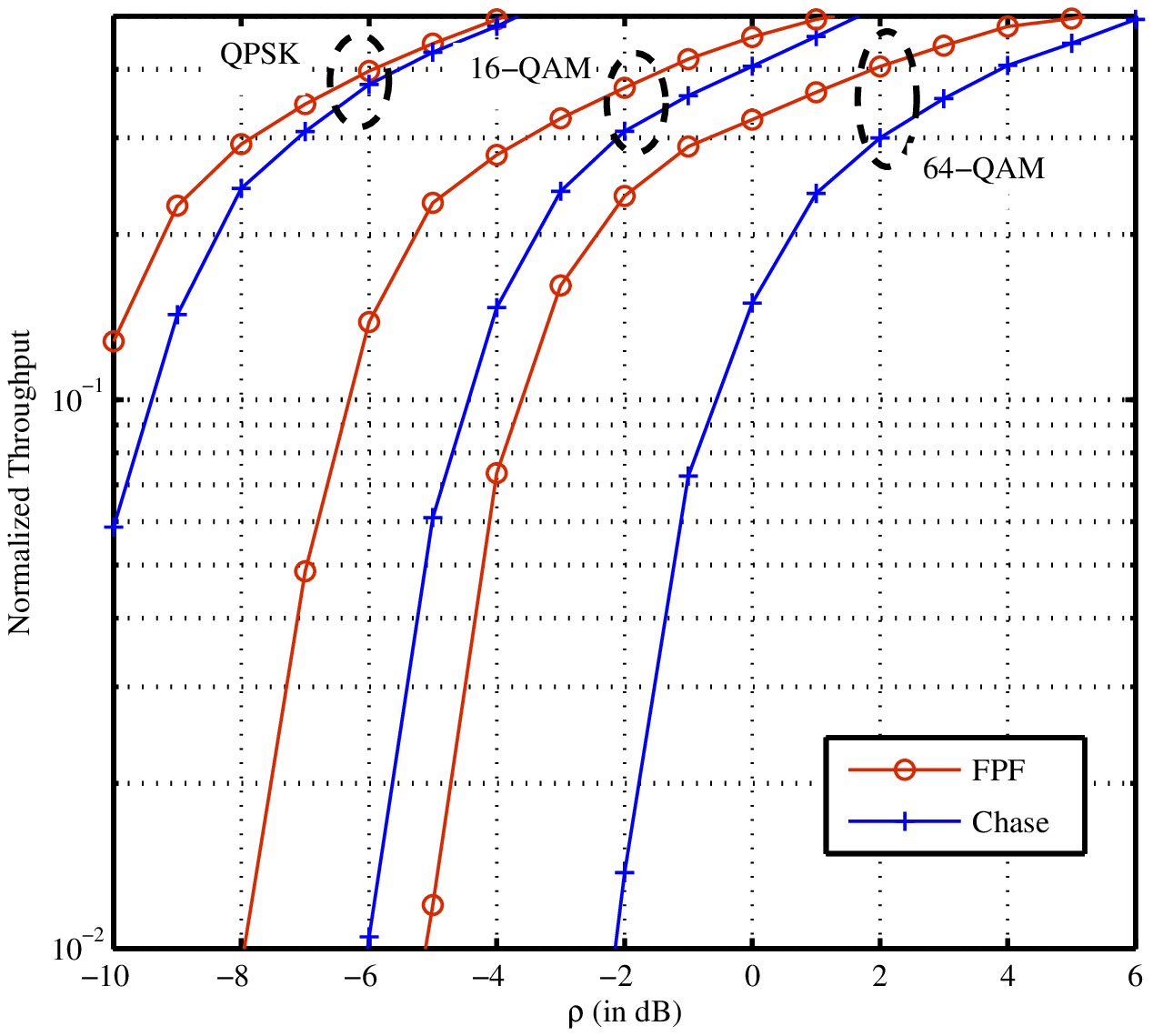}}
\subfloat[Partial Packet Feedback]{\label{ppf_plot}\includegraphics[scale = 0.55]{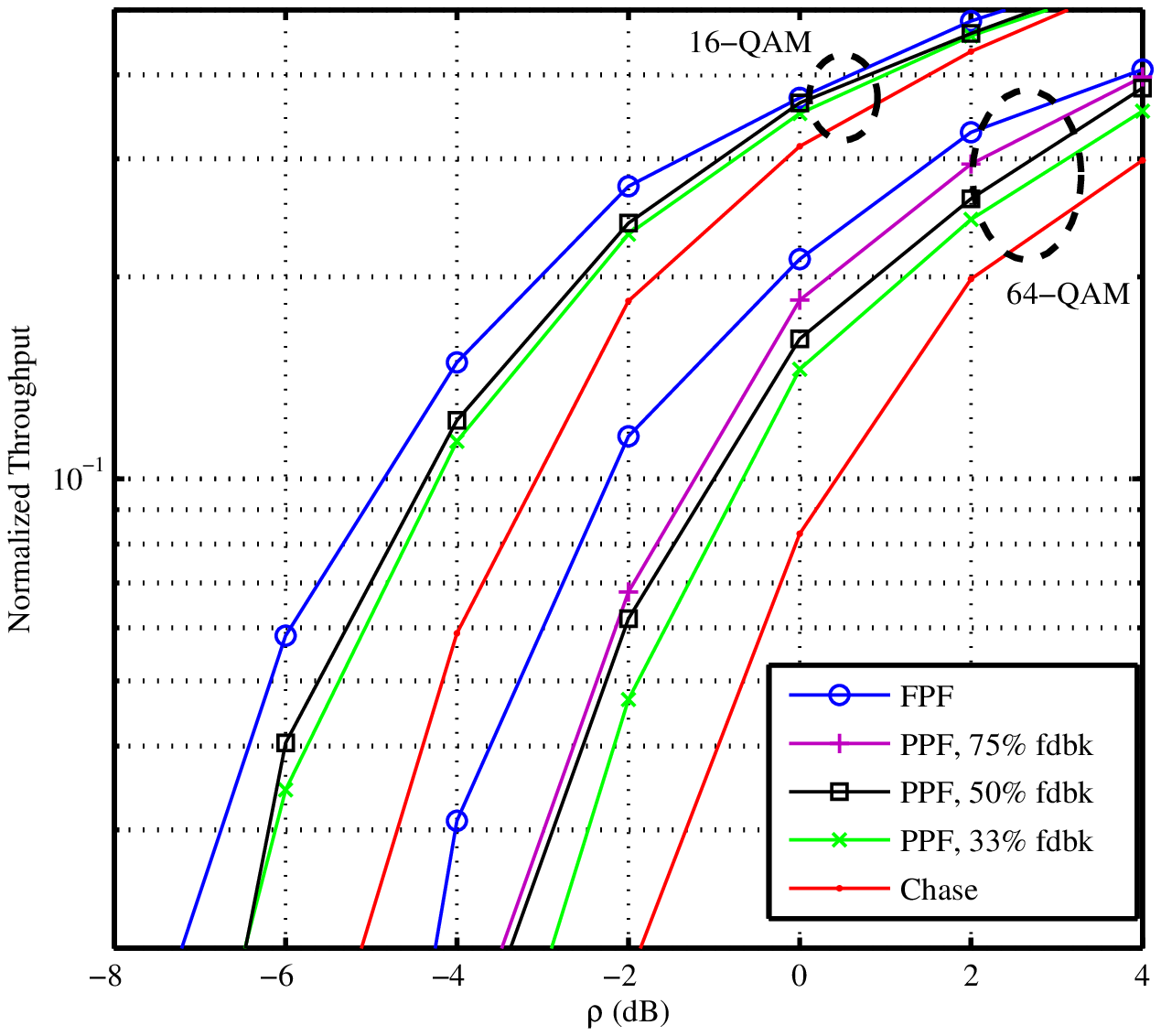}}
\newline
\subfloat[Partial COI]{\label{hsdpa_plot}\includegraphics[scale = 0.55]{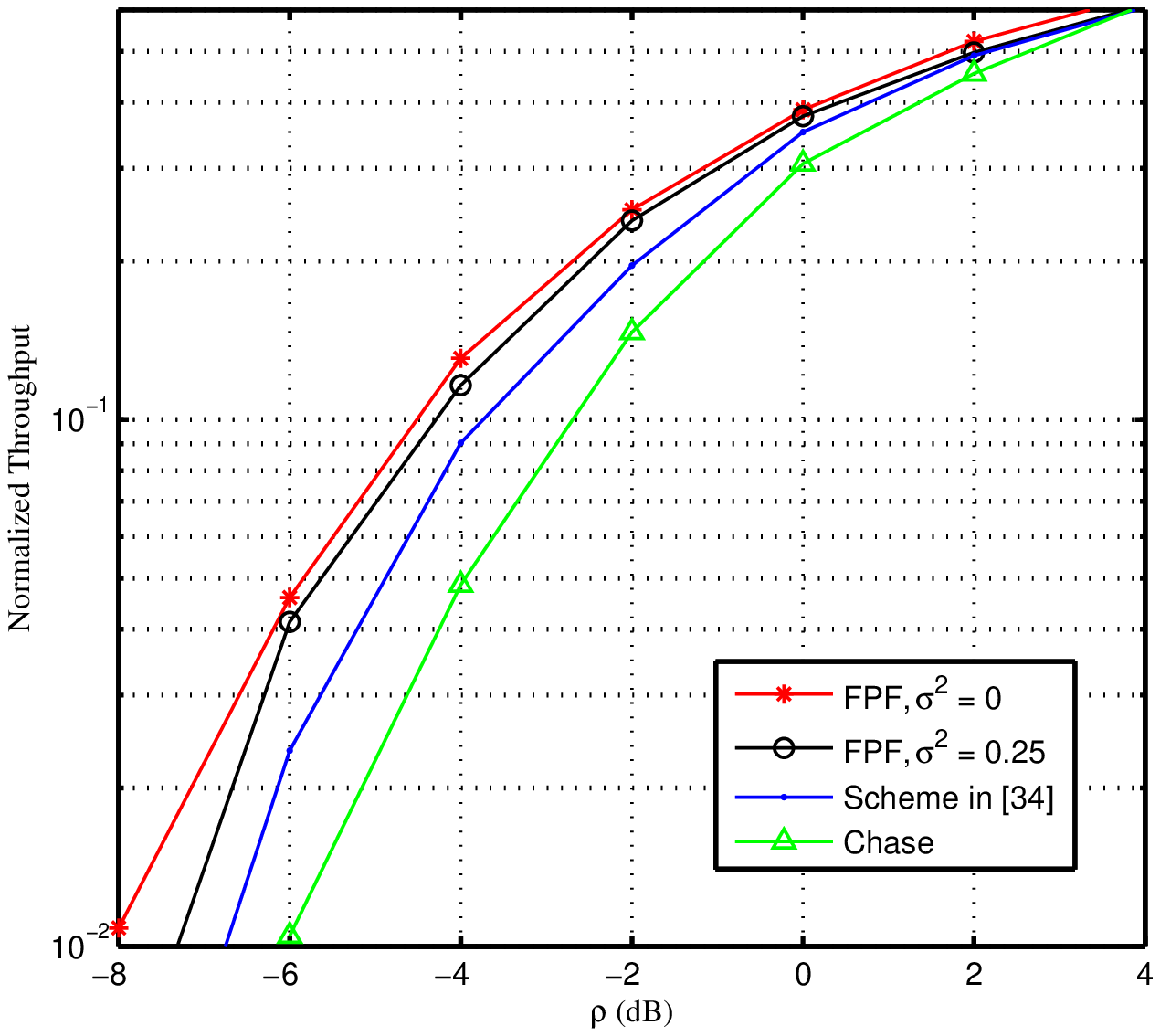}}
\subfloat[Quantized COI]{\label{quantized_plot}\includegraphics[scale = 0.55]{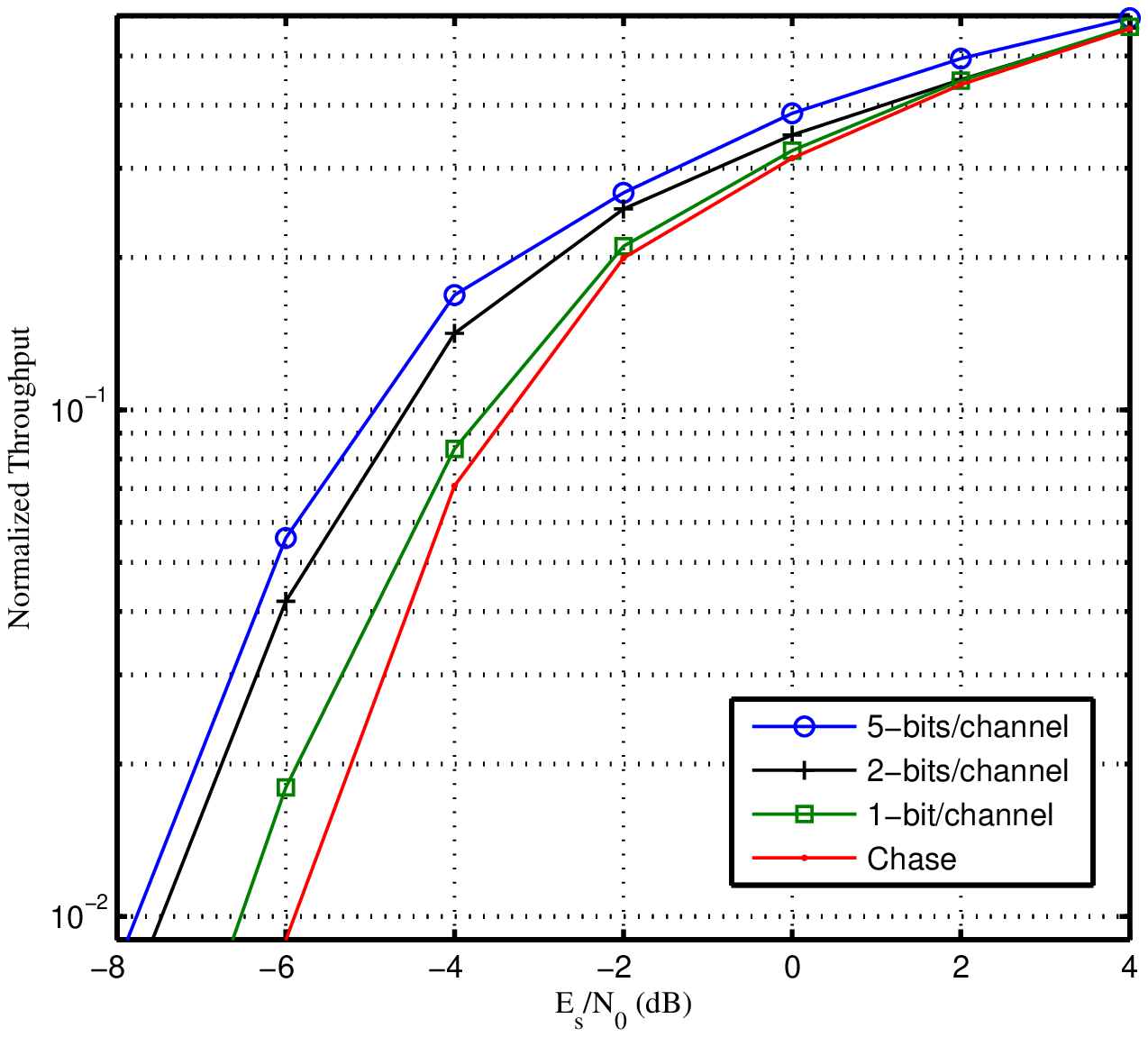}}
\caption{(a) Plot of variation of throughput for FPF and Chase combining with
channel SNR $\rho$ for $2 \times 2$ MIMO channel. The performance is compared
for QPSK, 16-QAM and 64-QAM constellations with $L_{\rm{info}}M = 2020$ bits.
(b) Plot of variation of throughput for PPF and Chase combining with channel SNR
$\rho$ for a SISO channel. The amount of PPF feedback is varied from $33\%$ to
$75\%$ of the total frame size. The performance is compared for 16-QAM and
64-QAM with $L_{\rm{info}} = 2020$ bits. (c) Plot of variation of throughput for
FPF with noisy feedback against Chase combining and HSDPA for a SISO channel.
The performance is compared for QPSK constellation with $L_{\rm{info}} = 3200$
bits. (d) Plot of variation of throughput for FPF with quantized COI feedback
against Chase combining for a SISO channel. The performance is compared for
16-QAM constellation with $L_{\rm{info}} = 2020$ bits.}
\label{plots}
\end{figure}

Fig. \ref{fpf_plot} compares the performance of FPF scheme with perfect COI
feedback against Chase combining for QPSK, 16-QAM, and 64-QAM constellations
over a $2 \times 2$ MIMO channel. The FEC code used for simulations is a 1/3
UMTS turbo code with eight decoding iterations. It is seen that most of the
gains from our proposed scheme are realized in the low SNR regime. The FPF
for QPSK displays gains of around 1 dB over Chase combining, and in 16-QAM, it gives an
improvement of about 2~dB over Chase combining. Furthermore, the gain increases
to 3~dB when the denser constellation of 64-QAM is chosen. It should be noted
that these gains have been realized directly at the packet level and not at the
bit level. This shows that with four retransmissions the power required at the
source can be halved with the inclusion of the proposed linear coding scheme.

In Fig. \ref{ppf_plot}, we plot the normalized throughput for PPF-PC for perfect
COI feedback against traditional Chase combining scheme for 16-QAM and 64-QAM
over a SISO channel. The amount of COI feedback symbols from the destination to
the source is varied from $33\%$ to $75\%$ of the total feedforward packet size.
For 64-QAM with a 1/3 UMTS turbo code and $L_{\rm info} = 2020$, the length of
the packet is $L = 1010$ symbols. Therefore the number of feedback symbols $T$
for 64-QAM is varied from $T = 337$ to $T = 757$. Again we can see the
improvements for 16-QAM and 64-QAM. Although the gains are smaller than the ones
for full packet feedback, they are still significant. It is actually interesting
to note that in 16-QAM most of the improvement in performance is reached with
only $50\%$ of COI feedback information. With $33\%$ COI feedback, PPF scheme
still shows an improvement of 1 dB over Chase for 16-QAM and a substantial
improvement of 2 dB for 64-QAM constellation.

Fig. \ref{hsdpa_plot} compares the normalized throughput for FPF with noisy COI
feedback against Chase combining and the scheme in~\cite{HSDPA}. It is seen that
even with a noise of $\sigma^2 = 0.25$ on the channel output feedback channel,
we see an improvement of about $0.5$~dB for the linear feedback scheme over the
scheme in \cite{HSDPA} in the low SNR regime. Furthermore this gain is realized
with only the addition of a very low complexity linear coder at source and
destination.

Finally, Fig. \ref{quantized_plot} compares the normalized throughput for FPF
with quantized COI feedback for 16-QAM constellation. At the destination, each
of the components -- inphase and quardature phase, are quantized using certain
number of bits and then mapped to the QAM constellation. It is seen that even
with 1-bit of quantization per phase leads to improvement in the performance
over the conventional Chase combining. However, most of the gains are obtained for
5-bits per phase of quantization.

\section{Conclusions}
In this paper, we have investigated a new hybrid-ARQ scheme that utilizes COI
feedback side-information from the destination.  This is motivated by trying to
close the performance gap between Chase combining and incremental redundancy
using feedback in order to leverage the implementation savings of a Chase
combining system \cite{Pal}. In normal Chase combining, packets are combined
using maximum ratio combining~(MRC); however, the proposed scheme incorporates
feedback by combining the packets using a linear feedback code for fading
channels with noisy feedback. Note that this also includes a new encoding step.
It was shown through Monte Carlo simulations that the post-processed SNR
performance of the linear feedback scheme greatly outperforms that of regular
MRC. In addition, since the code is built on linear operations, it adds little
complexity to the overall packet encoder and decoder assuming feedback side
information is present. The full hybrid-ARQ scheme was analyzed using two main
modes of operation: full packet feedback (FPF) in which the source was assumed
to have access to a noiseless/noisy version of the last received packet and
partial packet feedback (PPF) in which only a subset of the received symbol are
fed back to the source. Simulations show that the addition of feedback to
hybrid-ARQ greatly increases the performance and outperforms incremental
redundancy in most cases.

\begin{appendix}
\subsection{Proof of Lemma \ref{lem2}}
\begin{proof}
The encoding for perfect COI feedback using (\ref{betaeq}) can be written for
each $x[k]$ as
\begin{equation}
\label{enc1} x[k + 1] =  \phi^{-1}[k]e[k],
\end{equation}
where
\begin{eqnarray}
\label{enc3}e[k] & = & \theta -
\widehat{\theta}[k].
\end{eqnarray}
The operations at the decoder side can also be given by
\begin{eqnarray}
\label{dec1}\widehat{x}[k + 1] & = & \left(1 + \rho h^*[k + 1]h[k +
1]\right)^{-1}\rho h^*[k + 1]y[k + 1],\\
\label{dec2}\widehat{e}[k] & = &
\phi[k]\widehat{x}[k + 1],\\
\widehat{{\theta}}[k + 1] & = &
\label{dec3}\widehat{{\theta}}[k] + \widehat{e}[k].
\end{eqnarray}
For initialization purposes, it is assumed that $\widehat{\theta}[0] = 0$.  It can be seen from (\ref{enc3}) and (\ref{dec3}), that the error,
$e[k]$, for the symbol $\theta$ satisfies the relation
\begin{equation}
e[k + 1] = e[k] - \widehat{e}[k]. \label{errvec}
\end{equation}
Then, implementing (\ref{enc1}) and (\ref{errvec}), we can rewrite
$x[k+1]$ as
\begin{eqnarray*}
x[k + 1] & = & \phi^{-1}[k]\left(e[k - 1] - \widehat{e}[k -
1]\right)\\
& = & \phi^{-1}[k]\left(\phi[k - 1]x[k] -
\phi[k - 1]\widehat{x}[k]\right)\\
& = & \left(1 + \rho |h[k]|^2\right)^{1/2}(x[k] - \widehat{x}[k])\\
& = & \left(1 + \rho |h[k]|^2\right)^{-1/2}\left(x[k] -
\rho h^*[k]z[k]\right).
\end{eqnarray*}

According to (\ref{enc1}),
\begin{eqnarray*}
\theta - \widehat{{\theta}}[k]
& = & \phi[k]x[k + 1]\\
& = & \phi[k]\left(1 + \rho |h[k]|^2\right)^{-1/2}\left(x[k] -
\rho h^*[k]z[k]\right)\\
& = & \phi[k]\left(1 + \rho |h[k]|^2\right)^{-1/2}x[k] -
\rho\phi[k]\left(1 +
\rho |h[k]|^2\right)^{-1/2}h^*[k]z[k]\\
& = & |\phi[k]|^2\theta -
\rho|\phi[k]|^2\sum_{m =
1}^{k}\phi^{-1}[m - 1]h^*[m]z[m].
\end{eqnarray*}
Therefore, it follows that
\begin{equation*}
\widehat{{\theta}}[k] = \left(1 -
|\phi[k]|^2 \right)\theta + \rho
|\phi[k]|^2\sum_{m =
1}^{k}\phi^{-1}[m - 1]h^*[m]z[m].
\end{equation*}
\end{proof}
\subsection{Proof of Lemma \ref{lem4}}
\begin{proof}
We present the proof for $L = 1$. The generalization of it immediately follows.
For the $i^{th}$ spatial channel, we select the symbol $\theta_i$ from a square
QAM constellation consisting of $M_i[N] = 2^{NR_i}$ symbols. According to the
recursive definition in Lemma 2, the $i^{th}$ spatial signal is given as
\begin{equation}\label{l31}
\widetilde{x}_i[k] = \prod_{l = 1}^{k - 1}\frac{1}{\sqrt{1 + \rho
\widetilde{\lambda}^2_i[l]}}\theta_i - \rho\sum_{m = 1}^{k - 1}\left(\prod_{l =
m}^{k - 1} \frac{1}{\sqrt{1 +  \rho \widetilde{\lambda}^2_i[l]}}\right)
\widetilde{\lambda}_i[m] \widetilde{z}_i[m].
\end{equation}
Let
\begin{equation}\label{phi2}
\widetilde{\phi}_i[k] = \prod_{l = 1}^{k}\frac{1}{\sqrt{1 + \rho
\widetilde{\lambda}^2_i[l]}}, \quad \widetilde{\phi}_i[0] = 1.
\end{equation}

Now (\ref{l31}) can be rewritten as
\begin{equation}\label{l33}
\widetilde{x}_i[k] = \widetilde{\phi}_i[k - 1]\left( \theta_i - \rho\sum_{m =
1}^{k - 1}\frac{\widetilde{\lambda}_i[m] \widetilde{z}_i[m]}{\widetilde{\phi}_i[m - 1]}
\right).
\end{equation}
Based on (\ref{theta_ub}) which describes the unbiased estimation algorithm at
the destination,
\begin{equation}
\widehat{\theta}^u_i[N] =  \theta_i +
\rho \frac{\widetilde{\phi}_i^2[N]}{1 - \widetilde{\phi}_i^2[N]}\sum_{k =
1}^{N}\frac{\widetilde{\lambda}_i[k]}{\widetilde{\phi}_i[k -
1]}\widetilde{z}_i[k].
%& = & \theta_i - \widetilde{\phi}_i[k]\widetilde{x}_i[k + 1]
\end{equation}

Let
\begin{eqnarray}
\nonumber e^u_i[N] & = & \widehat{\theta}_i[N] - \theta_i.
\end{eqnarray}
Given channel realizations over blocklength $N$, $\{\bH
[k]\}_{k = 1}^N$, and a known $\theta_i$, the random variable $e^u_i[N]$ is just
a complex Gaussian random variable with conditional mean
\begin{eqnarray*}
\nonumber E \left[e^u_i[N] \Big{|} \{\bH[k]\}_{k = 1}^N,\theta_i\right] & = &
E\left[\rho \frac{\widetilde{\phi}_i^2[N]}{1 -
\widetilde{\phi}_i^2[N]}\sum_{k =
1}^{N}\frac{\widetilde{\lambda}_i[k]}{\widetilde{\phi}_i[k -
1]}\widetilde{z}_i[k] \Big{|}\{\bH[k]\}_{k = 1}^N, \theta_i\right]\\ & =
& 0.
\end{eqnarray*}
Similarly for the variance of $e^u_i[N]$, we obtain
\begin{eqnarray}
\nonumber \textrm{Var}\left(e^u_i[N]\Big{|}\{\bH[k]\}_{k = 1}^N,\theta_i\right) & = &
\textrm{Var}\left(\rho \frac{\widetilde{\phi}_i^2[N]}{1 -
\widetilde{\phi}_i^2[N]}\sum_{k =
1}^{N}\frac{\widetilde{\lambda}_i[k]}{\widetilde{\phi}_i[k -
1]}\widetilde{z}_i[k] \Big{|}\{\bH[k]\}_{k = 1}^N, \theta_i\right)\\ & = & \rho^2
\frac{\widetilde{\phi}_i^4[N]}{(1 - \widetilde{\phi}_i^2[N])^2} \sum_{k =
1}^{N}\frac{\widetilde{\lambda}^2_i[k]}{\widetilde{\phi}^2_i[k - 1]}.
\end{eqnarray}

The symbol $\theta_i$ is drawn from a square QAM constellation $\varTheta_i[N]$
given by,
\begin{equation}\label{m1}
\varTheta_i[N] = \sqrt{\alpha_i[N]}\left\{\pm 1
\pm 1j, \pm 1 \pm 3j,\cdots \cdots, \pm\left(\sqrt{M_i[N]}
- 1\right) \pm \left(\sqrt{M_i[N]}
- 1\right)j\right\},
\end{equation}
where the scaling factor $\alpha_i[N]$ satisfies the power constraint at the
source
\begin{equation}
E[|\theta_i|^2] = \frac{2}{3}(M_i[N] - 1)\alpha_i[N] = \rho.
\end{equation}

A correct decision about $\theta_i$ is made by the destination if the
error $e^u_i[N]$ falls within the square~($\Box_i[N]$) of length
$2\sqrt{\alpha_i[N]}$. Let %us denote
\begin{equation*}
P_e\left(\{\bH [k]\}_{k = 1}^N, \theta_i\right) = P\left(e^u_i[N] \notin
\Box_i[N]\Big{|}\{\bH [k]\}_{k = 1}^N, \theta_i\right).
\end{equation*}
Clearly,
\begin{eqnarray*}
\nonumber P_e\left(\{\bH [k]\}_{k = 1}^N, \theta_i\right) & \leq &
P\left(|\mathfrak{Re}(e^u_i[N])| > \sqrt{\alpha_i[N]}\Big{|}\{\bH [k]\}_{k = 1}^N,
\theta_i\right) +\\ & &
P\left(|\mathfrak{Im}(e^u_i[N])| > \sqrt{\alpha_i[N]}\Big{|}\{\bH [k]\}_{k = 1}^N,
\theta_i\right),
\end{eqnarray*}
where $\mathfrak{Re}(e^u_i[N])$ and $\mathfrak{Im}(e^u_i[N])$ denote the real
and imaginary part of $e^u_i[N]$ respectively. Using the identical distribution
of the real and imaginary components of the error $e^u_i[N]$, we get
\begin{equation*}
\nonumber P_e\left(\{\bH [k]\}_{k = 1}^N, \theta_i\right)
% & \leq &
%2P\left(|\mathfrak{Re}(e_i[k])| > \frac{\Delta_i[k]}{2}|\{\bH [l]\}_{l =
%0}^k,
%\theta_i\right)\\
\leq
4Q\left(\sqrt{\frac{{\alpha_i[N]}}{{\textrm{Var}\left(\mathfrak{Re}(e^u_i[N])\Big{|}
\{\bH [k]\}_{k = 1}^N,\theta_i\right)}}}\right).
\end{equation*}
Clearly,
\begin{equation*}
\textrm{Var}\left(\mathfrak{Re}(e^u_i[N])\Big{|}\{\bH [k]\}_{k = 1}^N, \theta_i
\right) = \frac{1}{2}\textrm{Var}\left(e^u_i[N]\Big{|}\{\bH [k]\}_{k = 1}^N,
\theta_i\right).
\end{equation*}

Therefore,
\begin{eqnarray}
\nonumber P_e\left(\{\bH [k]\}_{k = 1}^N, \theta_i\right) & \leq &
4Q\left(\sqrt{\frac{3\left(1 -
\widetilde{\phi}_i^2[N]\right)^2}{(M_i[N] - 1)\rho
\widetilde{\phi}_i^4[N]\sum_{k =
1}^{N}\frac{\widetilde{\lambda}^2_i[k]}{\widetilde{\phi}^2_i[k - 1]}}}\right)\\
& = & 4Q(\sqrt{a_i[N]}),
\end{eqnarray}
%Taking expectation on both sides, we get
%\begin{equation*}
%P_e \leq E \left[4Q(\sqrt{a_i[N]})\right],
%\end{equation*}
where
\begin{equation}\label{a1}
a_i[N] = {\frac{3\left(1 -
\widetilde{\phi}_i^2[N]\right)^2}{(M_i[N] - 1)\rho
\widetilde{\phi}_i^4[N]\sum_{k =
1}^{N}\frac{\widetilde{\lambda}^2_i[k]}{\widetilde{\phi}^2_i[k - 1]}}}.
\end{equation}
We next show that with probability 1, $a_i[N]$ increases
at least exponentially with $N$. From the definition of $\widetilde{\phi}_i[N]$
in (\ref{phi2}) we have $0 \leq \widetilde{\phi}_i[N] \leq
1,N \geq 0$. Also the definition implies that the sequence $\{
\widetilde{\phi}_i[N]\}_{N = 0}^{\infty}$ is a monotonically decreasing
sequence for arbitrary channel matrices. Hence by Theorem $3.14$
in~\cite{rudin}, the sequence $\{ \widetilde{\phi}_i[N]\}_{N = 0}^{\infty}$
converges. Also,
\begin{eqnarray}
E[\log_2\widetilde{\phi}_i[N]] & = & -\frac{1}{2}\sum_{k = 1}^N E \left[\log_2
\left(1 + \rho \widetilde{\lambda}_i^2[k]\right)\right]\\
\label{g2} & = & -\frac{N}{2}C_i
\end{eqnarray}

Using (\ref{g2}) and the strong law of large numbers~(SLLN), we know that
for any given $\epsilon > 0, \exists N_1$ such that
\begin{equation*}
P\left(\left|\frac{1}{N} \log_2 \widetilde{\phi}_i[N] +
\frac{1}{2}C_i \right| < \frac{\epsilon}{2}C_i\right) = 1 \quad
\forall N > N_1.
\end{equation*}

\noindent In particular,
\begin{equation}\label{b1}
P\left(\widetilde{\phi}_i[N] < 2^{-\frac{1}{2}N(1 -
\epsilon)C_i}\right)  =   1 \quad \forall N > N_1.
\end{equation}
By the almost sure convergence of $\{\widetilde{\phi}_i[N]\}_{N =
0}^{\infty}$ to zero, we can choose $N_2$ such that
\begin{equation}\label{b2}
P\left(1 - \widetilde{\phi}_i^2[N] > \frac{1}{\sqrt{3}}\right)  =   1 \quad
\forall N > N_2.
\end{equation}
Using SLLN again, we obtain that for a given $\epsilon > 0, \exists N_3$
such that
\begin{equation}\label{b3}
P\left(\sum_{k = 1}^N \widetilde{\lambda}^2_i[k] < \eta_i N(1 + \epsilon)
\right)  = 1 \quad \forall N > N_3,
\end{equation}
where $\eta_i = E[\widetilde{\lambda}^2_i[k]]$.
Substituting the bounds given by (\ref{b1}),\thinspace(\ref{b2}) and
(\ref{b3}) into the expression of $a_i[N]$ in (\ref{a1}), we obtain that
$\forall N > \max\{N_1, N_2, N_3\}$ with probability 1,
\begin{eqnarray*}
\nonumber a_i[N] & \geq &
\frac{1}{\rho}\frac{3\left(\frac{1}{\sqrt{3}}\right)^2}{2^{NR_i}2^{-N(1 -
\epsilon)C_i}\eta_i N(1 + \epsilon)}\\
& = & \frac{1}{\rho}\frac{2^{N\left((1 -
\epsilon)C_i - R_i\right)}}{\eta_i N(1 + \epsilon)}.
\end{eqnarray*}
The positive value $\epsilon$ also satisfies the inequality,
\begin{equation*}
\rho \eta_i N(1 + \epsilon) \leq 2^{\epsilon N C_i}, \quad \forall N >
N_4.
\end{equation*}
Clearly it follows that $\forall N > N_{\rm max}$
\begin{equation*}
a_i[N] \geq 2^{N\left((1 - 2\epsilon)C_i - R_i\right)},
\end{equation*}
where $N_{\rm max} = \max\{N_1, N_2, N_3, N_4\}$.

Thus, we have shown that with probability one, the input parameter of the
$Q$-function increases exponentially. Furthermore it is very well known that $Q$-function
decays exponentially and can be bounded by,
\begin{equation*}
Q(x) \leq \frac{1}{2}e^{-x^2/2}, \quad \forall x \geq 0.
\end{equation*}
From the above two equations we immediately obtain,
\begin{equation*}
P_e\left(\{\bH [k]\}_{k = 1}^N, \theta_i\right) \leq 2e^{-\frac{1}{2}2^{N\left((1 - 2\epsilon)C_i - R_i\right)}}.
\end{equation*}
Note that we can choose $\epsilon$ arbitrarily. Picking $\epsilon <
\displaystyle\min_{i = 1,2,\ldots,M}\textstyle\frac{1}{2}\left(1 - \frac{R_i}{C_i}\right)$ guarantees that the decay is
doubly exponential.
\end{proof}
\end{appendix}
\bibliographystyle{IEEEtran}
\bibliography{ARQ}

\begin{thebibliography}{10}
\providecommand{\url}[1]{#1}
\csname url@samestyle\endcsname
\providecommand{\newblock}{\relax}
\providecommand{\bibinfo}[2]{#2}
\providecommand{\BIBentrySTDinterwordspacing}{\spaceskip=0pt\relax}
\providecommand{\BIBentryALTinterwordstretchfactor}{4}
\providecommand{\BIBentryALTinterwordspacing}{\spaceskip=\fontdimen2\font plus
\BIBentryALTinterwordstretchfactor\fontdimen3\font minus
  \fontdimen4\font\relax}
\providecommand{\BIBforeignlanguage}[2]{{%
\expandafter\ifx\csname l@#1\endcsname\relax
\typeout{** WARNING: IEEEtran.bst: No hyphenation pattern has been}%
\typeout{** loaded for the language `#1'. Using the pattern for}%
\typeout{** the default language instead.}%
\else
\language=\csname l@#1\endcsname
\fi
#2}}
\providecommand{\BIBdecl}{\relax}
\BIBdecl

\bibitem{love2}
{D. J. Love, R. W. Heath, Jr., V. K. N. Lau, D. Gesbert, B. D. Rao, and M.
  Andrews}, ``An overview of limited feedback in wireless communication
  systems,'' \emph{IEEE Jour. Select. Areas in Commun.}, vol.~26, no.~8, pp.
  1341--1365, Oct. 2008.

\bibitem{SuMa06}
{H. Sun, J. Manton, and Z. Ding}, ``Progressive linear precoder optimization
  for {MIMO} packet retransmissions,'' \emph{IEEE Jour. Select. Areas in
  Commun.}, vol.~24, no.~3, pp. 448--456, Mar. 2006.

\bibitem{SaDi06}
H.~Samra and Z.~Ding, ``New {MIMO} {ARQ} protocols and joint detection via
  sphere decoding,'' \emph{IEEE Trans. Sig. Proc.}, vol.~54, no.~2, pp.
  473--482, Feb. 2006.

\bibitem{Schal1}
J.~Schalkwijk and T.~Kailath, ``A coding scheme for additive noise channels
  with feedback--{I}: {N}o bandwidth constraint,'' \emph{IEEE Trans. Info.
  Th.}, vol.~12, no.~2, pp. 172--182, April 1966.

\bibitem{Schal2}
J.~Schalkwijk, ``A coding scheme for additive noise channels with
  feedback--{I}{I}: {B}and-limited signals,'' \emph{IEEE Trans. Info. Th.},
  vol.~12, no.~2, pp. 183--189, April 1966.

\bibitem{butman}
S.~A. Butman, ``A general formulation of linear feedback communication systems
  with solutions,'' \emph{IEEE Trans. Info. Th.}, vol. IT-15, no.~3, pp.
  392--400, May 1969.

\bibitem{Caire1}
G.~Caire and D.~Tuninetti, ``The throughput of hybrid-{A}{R}{Q} protocols for
  the {G}aussian collision channel,'' \emph{IEEE Trans. Info. Th.}, vol.~47,
  no.~5, pp. 1971--1988, July 2001.

\bibitem{Caire3}
{C.F. Leanderson and G. Caire}, ``The performance of incremental redundancy
  schemes based on convolutional codes in the block-fading {G}aussian collision
  channel,'' \emph{IEEE Trans. on Wireless Comm.}, vol.~3, no.~3, pp. 843--854,
  May 2004.

\bibitem{Caire4}
S.~Sesia, G.~Caire, and G.~Vivier, ``Incremental redundancy hybrid {A}{R}{Q}
  schemes based on low-density parity-check codes,'' \emph{IEEE Trans.
  Commun.}, vol.~52, no.~8, pp. 1311--1321, Aug. 2004.

\bibitem{Love1}
J.~Wang, S.~Park, D.~J. Love, and M.~D. Zoltowski, ``Throughput delay tradeoff
  for wireless multicast using {H}ybrid-{A}{R}{Q} protocols,'' \emph{IEEE
  Trans. Commun.}, vol.~58, no.~9, pp. 2741--2751, Sep. 2010.

\bibitem{bruno}
B.~Clerckx, Y.~Zhou, and S.~Kim, ``Practical codebook design for limited
  feedback spatial multiplexing,'' in \emph{Proc. of IEEE Intl. Conf. on
  Communications}, May 2008, pp. 3982--3987.

\bibitem{SpPe04}
{Q. H. Spencer, C. B. Peel, A. L. Swindlehurst, and M. Haardt}, ``An
  introduction to the multi-user {MIMO} downlink,'' in \emph{IEEE
  Communications Magazine}, no.~10, Oct. 2004, pp. 60-- 67.

\bibitem{StBa04}
{G. L. Stuber, J. R. Barry, S. W. McLaughlin, Ye Li, M. A. Ingram, T. G.
  Pratt}, ``Broadband {MIMO-OFDM} wireless communications,'' in
  \emph{Proceedings of the IEEE}, no.~2, Feb. 2004, pp. 271--294.

\bibitem{PaDa08}
{S. Parkvall, E. Dahlman, A. Furuskar, Y. Jading, M. Olsson, S. Wanstedt, and
  K. Zangi}, ``{LTE}-{A}dvanced - evolving {LTE} towards {IMT}-{A}dvanced,'' in
  \emph{Vehicular Technology Conference}, Sep. 2008, pp. 1--5.

\bibitem{Lin}
S.~Lin and D.~Costello, \emph{Error Control Coding: Fundamentals and
  Applications}.\hskip 1em plus 0.5em minus 0.4em\relax Englewood Cliffs, NJ:
  Prentice-Hall, 1983.

\bibitem{Chase1}
D.~Chase, ``Code combining--{A} maximum-likelihood decoding approach for
  combining an arbitrary number of noisy packets,'' \emph{IEEE Trans. on
  Comm.}, vol.~33, no.~5, pp. 385--393, May 1985.

\bibitem{WozHor2}
J.~M. Wozencraft and M.~Horstein, ``Coding for two-way channels,'' \emph{Tech
  Rep. 383, Res. Lab Electron., MIT}, January 1961.

\bibitem{Caire2}
D.~Tuninetti and G.~Caire, ``The throughput of some wireless multiaccess
  systems,'' \emph{IEEE Trans. Info. Th.}, vol.~48, no.~10, pp. 2773--2785,
  Oct. 2002.

\bibitem{Val}
T.~Ghanim and M.~C. Valenti, ``The throughput of hybrid-{A}{R}{Q} in block
  fading under modulation constraints,'' in \emph{40th Annual Conf. on Info.
  Sciences and Systems}, March 2006, pp. 253--258.

\bibitem{Jung}
P.~Jung, J.~Plechinger, M.~Doetsch, and F.~M. Berens, ``A pragmatic approach to
  rate compatible punctured turbo-codes for mobile radio applications,'' in
  \emph{Proc. 6th Int. Conf. on Advances in Comm. and Control}, June 1997.

\bibitem{Rowitch}
D.~N. Rowitch and L.~Milstein, ``On the performance of hybrid
  {F}{E}{C}/{A}{R}{Q} systems using rate compatible punctured turbo
  ({R}{C}{P}{T}) codes,'' \emph{IEEE Trans. Commun.}, vol.~48, no.~6, pp.
  948--959, June 2000.

\bibitem{Aci}
O.~F. Acikel and W.~E. Ryan, ``Punctured turbo codes for
  {B}{P}{S}{K}/{Q}{P}{S}{K} channels,'' \emph{IEEE Trans. Commun.}, vol.~47,
  no.~9, pp. 1315--1323, Sep. 1999.

\bibitem{Hag}
J.~Hagenauer, ``Rate compatible punctured convolutional codes ({R}{C}{P}{C}
  codes) and their applications,'' \emph{IEEE Trans. Commun.}, vol.~36, no.~4,
  pp. 389--400, Apr. 1988.

\bibitem{Ha}
J.~Ha, J.~Kim, and S.~W. McLaughlin, ``Rate-compatible puncturing of
  low-density parity-check codes,'' \emph{IEEE Trans. Info. Th.}, vol.~50,
  no.~11, pp. 2824--2836, Nov. 2004.

\bibitem{Varn}
E.~Soljanin, N.~Varnica, and P.~Whiting, ``Incremental redundancy hybrid
  {A}{R}{Q} with {L}{D}{P}{C} and raptor codes,'' \emph{IEEE Transactions on
  Information Theory}, Submitted September 2005.

\bibitem{Lee}
C.~Lee and W.~Gao, ``Rateless-coded hyrbid {A}{R}{Q},'' in \emph{Proc. Int.
  Conf. on Info. and Comm. Systems}, Dec. 2007, pp. 1--5.

\bibitem{SuDi06}
{H. Sun, and Z. Ding}, ``Robust precoder design for {MIMO} packet
  retransmissions over imperfectly known flat-fading channels,'' in \emph{Proc.
  IEEE Int. Conf. on Commun.}, June 2006, pp. 3287--3292.

\bibitem{SuDi07}
H.~Sun and Z.~Ding, ``Iterative transceiver design for {MIMO ARQ}
  retransmissions with decision feedback detection,'' \emph{IEEE Trans. Sig.
  Proc.}, vol.~55, no.~7, pp. 3405--3416, July 2007.

\bibitem{KiKa09}
{J. W. Kim, C. G. Kang. B. J. Kwak, and D. S. Kwon}, ``Design of a codebook
  structure for a progressively linear pre-coded closed-loop {MIMO} hybrid
  {ARQ} system,'' in \emph{Proc. IEEE Int. Conf. on Commun.}, June 2009, pp.
  1--5.

\bibitem{AsPo10}
{A.-n. Assimi, C. Poulliat, and I. Fijalkow}, ``Phase-precoding without {CSI}
  for packet retransmissions over frequency-selective channels,'' \emph{IEEE
  Trans. Commun.}, vol.~58, no.~3, pp. 975--985, Mar. 2010.

\bibitem{Matsu}
T.~Matsushima, ``Adaptive incremental redundancy coding and decoding schemes
  using feedback information,'' in \emph{Proc. of IEEE Info. Theory Workshop},
  Oct. 2006, pp. 189--193.

\bibitem{Shea}
J.~M. Shea, ``Reliability-based hybrid {A}{R}{Q},'' \emph{IEE Electronic
  Letters}, vol.~38, no.~13, June 2002.

\bibitem{ZaDa11}
Z.~Chance and D.~J. Love, ``Concatenated coding for the {AWGN} channel with
  noisy feedback,'' \emph{IEEE Trans. Info. Th.}, vol.~57, no.~10, pp.
  6633--6649, Oct. 2011.

\bibitem{HSDPA}
{European Telecommunications Standards Institute}, ``Universal mobile
  telecommunications system ({U}{M}{T}{S}): {M}ultiplexing and channel coding
  ({F}{D}{D}),'' \emph{3GPP TS 25.212. version 6.6.0, release 6}, September
  2005.

\bibitem{liu1}
{J. Liu, N. Elia, and S. Tatikonda}, ``Capacity-achieving feedback scheme for
  flat fading channels with channel state information,'' in \emph{Proc. of the
  {A}merican {C}ontrol {C}onference}, vol.~4, June 2004, pp. 3593--3598.

\bibitem{El04}
N.~Elia, ``When {B}ode meets {S}hannon: {C}ontrol-oriented feedback
  communication schemes,'' \emph{IEEE Trans. Auto. Contr.}, vol.~49, no.~9, pp.
  1477--1488, Sep. 2004.

\bibitem{elias}
P.~Elias, ``Channel capacity without coding,'' in \emph{Quarterly progress
  report, MIT RLE}, 1956, pp. 90--93.

\bibitem{gallager}
R.~G. Gallager, \emph{Information {T}heory and {R}eliable
  {C}ommunication}.\hskip 1em plus 0.5em minus 0.4em\relax John Wiley, 1968.

\bibitem{honig}
W.~Santipach and M.~L. Honig, ``Asymptotic performance of {MIMO} wireless
  channels with limited feedback,'' in \emph{Proc. IEEE Mil. Comm. Conf.},
  vol.~1, Oct. 2003, pp. 141--146.

\bibitem{honig2}
------, ``Capacity of a multiple-antenna fading channel with a quantized
  precoding matrix,'' \emph{IEEE Trans. Info. Th.}, vol.~55, no.~3, pp.
  1218--1234, Mar. 2009.

\bibitem{love3}
{D. J. Love, R. W. Heath, Jr. and T. Strohmer}, ``Grassmanian beamforming for
  multiple-input multiple-output wireless systems,'' \emph{IEEE Trans. Info.
  Th.}, vol.~49, no.~10, pp. 2735--2747, Oct. 2003.

\bibitem{Mukka}
K.~K. Mukkavilli, A.~Sabharwal, E.~Erkip, and B.~Aazhang, ``On beamforming with
  finite rate feedback in multiple antenna systems,'' \emph{IEEE Trans. Info.
  Th.}, vol.~49, no.~10, pp. 2562 -- 2579, Oct. 2003.

\bibitem{UMTS}
M.~C. Valenti and J.~Sun, ``The {U}{M}{T}{S} turbo code and an efficient
  decoder implementation suitable for software defined radios,'' \emph{Int.
  Journal Wireless Info. Networks}, vol.~8, no.~4, pp. 203--216, Oct. 2001.

\bibitem{Pal}
P.~Frenger, S.~Parkvall, and E.~Dahlman, ``Performance comparison of
  {H}{A}{R}{Q} with {C}hase combining and incremental redundancy for
  {H}{S}{D}{P}{A},'' in \emph{Proc. Vehicular Tech. Conf.}, vol.~3, Oct. 2001,
  pp. 1829--1833.

\bibitem{rudin}
W.~Rudin, \emph{Principles of {M}athematical {A}nalysis}.\hskip 1em plus 0.5em
  minus 0.4em\relax Mc-Graw Hill, 1976.

\end{thebibliography}

\end{document}